\documentclass[twoside,english]{elsarticle}
\usepackage[T1]{fontenc}
\usepackage[latin9]{inputenc}
\usepackage[letterpaper]{geometry}
\usepackage{float}
\usepackage{url}
\usepackage{amsthm}
\usepackage{amsbsy}
\usepackage{amstext}
\usepackage{graphicx}

\makeatletter

\providecommand{\tabularnewline}{\\}
\floatstyle{ruled}
\newfloat{algorithm}{tbp}{loa}
\floatname{algorithm}{Algorithm}

\theoremstyle{plain}
\newtheorem{thm}{Theorem}
  \theoremstyle{definition}
  \newtheorem{defn}[thm]{Definition}
  \theoremstyle{plain}
  \newtheorem{lem}[thm]{Lemma}
  \theoremstyle{remark}
  \newtheorem*{rem*}{Remark}

  \theoremstyle{plain}
  \newtheorem*{fact*}{Fact}
  \theoremstyle{remark}
  \newtheorem{rem}[thm]{Remark}
  \theoremstyle{plain}
  \newtheorem{cor}[thm]{Corollary}
  \theoremstyle{plain}
  \newtheorem{fact}[thm]{Fact}

\journal{Discrete Applied Mathematics}
\usepackage{amsmath}
\usepackage[ruled]{algorithm}
\usepackage[noend]{algpseudocode}

\usepackage{graphics,graphicx} 
\usepackage{pstricks,pst-node}
\usepackage{wrapfig}
\usepackage{subfigure}

\newcommand{\npc}{$NP$-complete}
\newcommand{\npcness}{$NP$-completeness}
\newcommand{\np}{$NP$}

\newcommand{\ksat}{k-SAT}

\newcommand{\threesat}{3-SAT}

\makeatother

\usepackage{babel}

\begin{document}

\title{Estimating Satisfiability}

\author[liafa]{Yacine~Boufkhad\fnref{fn1}\corref{cor1}}

\ead{boufkhad@liafa.jussieu.fr}

\author[liafa]{Thomas~Hugel\fnref{fn1}}

\ead{thomas.hugel@liafa.jussieu.fr}

\fntext[fn1]{Partially supported by INRIA project Gang}

\cortext[cor1]{Corresponding author}

\address[liafa]{LIAFA, CNRS UMR 7089, Université Denis Diderot Paris 7, Case 7014,
F-75205 Paris Cedex 13}

\address{}
\begin{abstract}
The problem of estimating the proportion of satisfiable instances
of a given CSP (constraint satisfaction problem) can be tackled through
weighting. It consists in putting onto each solution a non-negative
real value based on its neighborhood in a way that the total weight
is at least $1$ for each satisfiable instance. We define in this
paper a general weighting scheme for the estimation of satisfiability
of general CSPs. First we give some sufficient conditions for a weighting
system to be correct. Then we show that this scheme allows for an
improvement on the upper bound on the existence of non-trivial cores
in \threesat{} obtained by \citet{Maneva2008} to $4.419$. Another
more common way of estimating satisfiability is ordering. This consists
in putting a total order on the domain, which induces an orientation
between neighboring solutions in a way that prevents circuits from
appearing, and then counting only minimal elements. We compare ordering
and weighting under various conditions.\end{abstract}
\begin{keyword}
Constraint Satisfaction Problem \sep Satisfiability \sep First Moment
Method 
\end{keyword}
\maketitle

\section{Introduction}

Constraint satisfaction problems cover a large variety of problems
that arise in many areas of combinatorial optimization. They are central
in complexity theory because they are \npc{} and also because one
particular case - satisfiability of boolean formulas - was the first
problem to be identified in this class. In general, they consist in
defining constraints on a set of \emph{variables }taking their \emph{values
}in a given finite domain. \emph{Constraints} specify which combinations
of values assigned to subsets of variables are allowed (or dually
are forbidden). A \emph{solution} is a valuation (i.e. the assignment
of a value to each variable) that does not violate any constraint.
The satisfiability problem is the following: given an instance, decide
the existence of a solution for it.

Besides the design of algorithms for solving these problems, the research
of structural properties for these problems has attracted much attention
in the recent years. In particular, the empirical evidence of the
existence of a threshold (rigorously established in some particular
cases) in the satisfiability of some classes of CSPs has opened a
field of research: attempts are made to rigorously establish the existence
and the location of this threshold. This involves estimating the proportion
of satisfiable instances in a given set of instances. The \npcness{}
of these problems in general makes it difficult to determine whether
a given instance is satisfiable; that may explain why direct counting
of satisfiable instances is currently unfeasible. However, precisely
because these problems are in \np{}, it is easy to determine whether
some instance is satisfied by a given valuation and then to count
the formulas satisfied by this valuation. Thus counting couples (formulas,
solutions) is only accessible starting from a solution; moreover,
given a solution, it is not complicated to investigate also its immediate
neighborhood. But even at a distance of 2, i.e.~with neighbors of
neighbors, calculations become quite complicated (see \citet{Kirousis1998}).
This fact imposes a strong restriction on the design of both estimation
techniques studied: they can only make use of local information. We
shall refer to this as \emph{the locality condition}.

Using one of the most popular techniques in the probabilistic method
(cf.~\citet{Alon1992}), namely the first moment method, it is possible
to bound from above the probability of satisfiability. The implementation
of the first moment method makes use of Markov's inequality; one needs
to define a non-negative random variable $X$ that must be at least
$1$ for a satisfiable formula (we call that a \emph{correct }random
variable). Ideally, $X$ should be as small as possible; in other
words, it should be $0$ for unsatisfiable instances and as close
to $1$ as possible for satisfiable ones (if $X$ is $1$ for every
satisfiable instance and $0$ for every unsatisfiable instance then
we get the exact probability of satisfiability). The most straightforward
candidate for $X$ is simply the number of solutions. The method consists
in counting for every valuation the number of instances that are satisfied
by it and then summing up over all valuations. But since the number
of solutions is generally too large, the method over-estimates the
proportion of satisfiable instances. 

Many techniques have been developed to overcome this difficulty in
various types of CSPs, Satisfiability of CNF formulas (\citet{Kamath1995,Dubois1997,Kirousis1998,Dubois2000,Dubois2003,Boufkhad2005,Kaporis2007,Diaz2009,Boufkhad2010}),
3-Coloring of graphs (\citet{Achlioptas1999}), Binary CSPs (\citet{Achlioptas1997,Achlioptas2001b})\ldots{}
Most of these methods share a common point: they count minimal elements
under some partial order over solutions. We will refer to this method
as \emph{solution selection through a partial ordering} or for short
\emph{ordering}. Due to the locality condition, the partial order
must be locally computable (i.e. must depend only on the immediate
neighbors of the considered solution). Two solutions of some instance
are neighbors if they disagree only on the value taken by one variable.
Both solutions may be ordered using a predetermined order on the values
for this particular variable in this particular instance. Finally
we count only those solutions having minimal values for all their
variables with respect to their neighbors. 

Recently, \citet{Maneva2007} introduced a novel approach for the
boolean satisfiability problem consisting in weighting partial valuations
and solutions depending on their neighborhood. While not originally
intended to estimate the proportion of satisfiable instances (but
rather to analyze some properties of Belief Propagation algorithms),
it was though specifically used by \citet{Maneva2008} to estimate
the probability of existence of non-trivial cores in random \threesat{}
instances. The existence of non-trivial cores contains an important
information on the structure of the space of solutions; moreover it
is related to the clustering that has been proved to exist in \ksat{}
for $k\geq9$ \citet{Achlioptas2006}. \citet{Maneva2008} show that
in the \threesat{} instances, non-trivial cores do not exist for
ratios of clauses to variables greater that $4.453$. To do so they
use \emph{valid partial valuations} (i.e. satisfying some properties
related to boolean satisfiability) and weight them according to their
values and their neighborhood.

\section{Overview of Results}

Our first result consists in giving some sufficient conditions to
make a weighting scheme correct for the estimation of satisfiability
on general CSPs (theorem \ref{Weight-Conservation-Theorem}, \emph{Weight
Conservation Theorem}). Then we propose a general weighting scheme
obeying these conditions (theorem \ref{thm:dispatchers-are-good}).
This scheme is based on:
\begin{enumerate}
\item a \emph{weighting seed} that expresses the relative importance of
each value with respect to a variable and an instance; the seed is
such that if all valuations were solutions, then their total weight
would be exactly $1$;
\item a \emph{dispatching function }expressing how the weights of forbidden
valuations are dispatched among solutions to insure that counting
weighted solutions will yield at least $1$ for any satisfiable instance.
\end{enumerate}
We will refer to this method as \emph{solution weighting }or\emph{
}for short \emph{weighting.} 

In theorem \ref{thm:Maneva}, we show that the estimation of satisfiability
used by \citet{Maneva2008} can be improved upon by using a weighting
scheme based on a $3$-valued CSP and obeying the conditions of our
\emph{Weight Conservation Theorem} (which shows that these conditions
are somehow relevant). Thanks to this weighting system, we improve
on the upper bound on the existence of non-trivial cores to $4.419$.
We completely reuse the proof of \citet{Maneva2008} for our new weighting
system, showing that the improvement on the value of the bound is
indeed due to a better weighting system.

Till now the only way to compare ordering and weighting was to compute
the estimations of satisfiability obtained by each of them on a certain
set of instances and to choose the best one. We give some results
comparing these two ways of estimating satisfiability in the following
cases:
\begin{itemize}
\item weighting and ordering can be instance dependent when such syntactic
properties as the number of occurrences of variables and values etc.~can
guide the design of weighting functions and orderings. We show that
in the general case where the weighting function is instance dependent
and when the weighting is \emph{homogeneous} (i.e. when weighting
seeds and dispatching functions are equal), weighting cannot be better
than a well chosen instance dependent ordering (theorem \ref{thm:orientation-better});
\item in the case where ordering and weighting are instance independent
(which is the case of problems where the values are indistinguishable
like graph coloring for example) and in the case of sets of instances
closed under value renaming (which is the case of almost all sets
of instances considered in the literature), we show that weighting
and ordering are equivalent on average (theorem \ref{thm:set-equality}).
\end{itemize}

\section{Framework}

A \emph{CSP} (Constraint Satisfaction Problem) is a triple $F=\left\langle X,D,C\right\rangle $
where $X$ is a set of \emph{variables} taking their values in the
same finite \emph{domain} $D$ of values, and $C$ is a set of \emph{constraints}.
A \emph{constraint} is a couple $\left\langle \boldsymbol{x},R\right\rangle $
where $\boldsymbol{x}\in X^{k}$ and $R\subseteq D^{k}$ for some
integer $k$. $R$ is interpreted as the tuples of allowed values.
A \emph{valuation} is a vector $v\in D^{X}$; access to coordinate
$x\in X$ of $v$ will be denoted as $v\left(x\right)$. It satisfies
some constraint $\left\langle \left(x_{1},x_{2},...,x_{k}\right),R\right\rangle $
iff $\left(v\left(x_{1}\right),v\left(x_{2}\right),...,v\left(x_{k}\right)\right)\in R$.
A valuation is said to be a \emph{solution} of a CSP instance iff
it satisfies all of its constraints.

We consider some sets $\mathcal{F}$ of CSP instances sharing the
same set $X$ of variables and the same domain $D$.  In the rest
of the paper $n=\left|X\right|$ denotes the number of variables,
$d=\left|D\right|$ the size of the domain. Given a CSP instance $F$,
let $\mathcal{S}\left(F\right)$ denotes the set of its solutions.

We are interested in the neighborhood of valuations. Given a valuation
$v$ and $a\in D$, we define $v_{x\gets a}$ as the valuation obtained
from $v$ by changing the value of $x$ to $a$ (including the case
when already $a=v\left(x\right)$). Given a variable $x$, two solutions
are called \emph{$x$-adjacent} if they agree on all variables but
$x$: in other words $\sigma$ and $\tau$ are $x$-adjacent iff $\tau=\sigma_{x\gets\tau\left(x\right)}$.
Note that for each variable $x$, $x$-adjacency is an equivalence
relation on solutions. Bringing together the $x$-adjacency relations
with respect to every variable and removing the loops $\left(\sigma,\sigma\right)$
we get an non-oriented graph on $\mathcal{S}\left(F\right)$ that
we call \emph{solutions network}. Let $N_{F}\left(\sigma,x\right)$
denote the equivalence class of $\sigma$ under $x$-adjacency (i.e.
the neighborhood of $\sigma$ for variable $x$); note that $N_{F}\left(\sigma,x\right)$
is a clique for $x$-adjacency. Such a clique will play a central
role in our weighting system. We are also interested in the different
values that $x$ takes in this equivalence class, so we define $A_{F}\left(\sigma,x\right)=\left\{ \tau\left(x\right)\right\} _{\tau\in N_{F}\left(\sigma,x\right)}$.
For example in figure \ref{fig:basic-examples-orientations}, solutions
$ab$ and $aa$ are $y$-adjacent, $N_{F}\left(ab,y\right)=\left\{ ab,aa\right\} $
and $A_{F}\left(ab,y\right)=\left\{ b,a\right\} $.

Most of the results in this paper apply to any set of solutions, regardless
of which CSP instance has generated them. The sole solutions network
can be thought of as the input of the problem. However it should be
borne in mind that weightings and orderings cannot be defined using
the global knowledge of the whole set of solutions, because of the
locality condition: one can only count instances having a given solution,
and for each instance the solutions that are neighbors of this solution
(rather than all solutions of a given instance). A convenient way
to visualize this limitation is to imagine a network of processors
(a processor representing a solution) where each processor has knowledge
of its neighbors only and must compute from this knowledge its own
weight or determine the orientation with respect to its neighbors.

\section{Weighting of Solutions}

First we define a weighting system for all valuations (solutions or
not) which sums up to $1$. Then we give sufficient conditions on
a weighting system on solutions only, such that a transfer between
this weighting system and the previous one may be possible. Doing
this we establish a general framework for putting weights onto solutions,
and use it to derive two particular weighting systems: the first one
addresses general CSPs and the second one is built to improve on the
weighting system introduced by \citet{Maneva2007,Maneva2008,Ardila2009a}.
The purpose of such a transfer is to estimate the global weight in
the weighting system on solutions by means of the global weight of
the weighting system on all valuations (which is easier to compute).

\subsection{Weighting Seeds}
\begin{defn}
\label{def:U}For a CSP $F=\left\langle X,D,C\right\rangle $ a \emph{weighting
seed} is a function $s_{F}\,:\, X\times D\rightarrow\boldsymbol{R}^{+}$.
We say that $s_{F}$ is \emph{unitary }iff $\forall x\in X,\sum_{a\in D}s_{F}\left(x,a\right)=1$.

Now we define the \emph{unladen weight }of any valuation $v$ (solution
or not) with respect to some weighting seed $s_{F}$ as : \begin{equation}
U_{F}\left(v\right)=\prod_{x\in X}s_{F}\left(x,v\left(x\right)\right)\enskip.\end{equation}

As for the \emph{actual weight} of a solution, we want to take into
account the neighborhood of the solution, so we put the weight $w_{F}\left(\sigma,x\right)$
on each variable $x$ of solution $\sigma$. We will see later how
to build $w_{F}$ from $s_{F}$.  

The actual weight of a solution is: \begin{equation}
W_{F}\left(\sigma\right)=\prod_{x\in X}w_{F}\left(\sigma,x\right)\enskip.\end{equation}

By extension, the weight of a set $S$ of solutions is:\begin{equation}
W_{F}\left(S\right)=\sum_{\sigma\in S}W_{F}\left(\sigma\right)\enskip.\end{equation}
\end{defn}
\begin{lem}
\label{lem:unitary}If the weighting seed $s_{F}$ is unitary, then
the total unladen weight of all valuations is $1$: $\sum_{v\in D^{X}}U_{F}\left(v\right)=1$.\end{lem}
\begin{proof}
\begin{eqnarray*}
\sum_{v\in D^{X}}U_{F}\left(v\right) & = & \sum_{v\in D^{X}}\prod_{x\in X}s_{F}\left(x,v\left(x\right)\right)\\
 & = & \prod_{x\in X}\sum_{a\in D}s_{F}\left(x,a\right)\\
 & = & \prod_{x\in X}1\\
 & = & 1\enskip.\end{eqnarray*}

\end{proof}
This weight $U_{F}$ is indeed simple to handle. The purpose is now
to connect it with $W_{F}$. Just as we defined weights $W_{F}$ of
solutions in a product form variable per variable, so shall we build
our transfer system.

\subsection{Decomposers}
\begin{defn}
\label{def:T} We say that $w_{F}$ is \emph{decomposable} by a family
$\left(\delta_{F,\sigma,x,a}\right)$ iff for all solution $\sigma$
of $F$ and all variable $x$, $w_{F}\left(\sigma,x\right)=\sum_{a\in D}\delta_{F,\sigma,x,a}$.
Such a family will be referred to as a \emph{decomposer}. We define
onto it the following transfer quantities between a solution $\sigma$
and a valuation $v$: \begin{eqnarray}
T_{F,\sigma\to v} & = & \prod_{x\in X}\delta_{F,\sigma,x,v\left(x\right)}\enskip.\end{eqnarray}
\end{defn}
\begin{lem}
\label{lem:Transfer}(Transfer lemma). Let $F$ be a CSP instance
and $\sigma$ any of its solutions. If $w_{F}$ is decomposable by
family $\left(\delta_{F,\sigma,x,a}\right)$, then\begin{eqnarray}
W_{F}\left(\sigma\right) & = & \sum_{v\in D^{X}}T_{F,\sigma\to v}\enskip.\end{eqnarray}
\end{lem}
\begin{proof}
It is sufficient to expand the weight of a solution as follows: \begin{eqnarray*}
W_{F}\left(\sigma\right) & = & \prod_{x\in X}w_{F}\left(\sigma,x\right)\\
 & = & \prod_{x\in X}\sum_{a\in D}\delta_{F,\sigma,x,a}\\
 & = & \sum_{v\in D^{X}}\prod_{x\in X}\delta_{F,\sigma,x,v\left(x\right)}\\
 & = & \sum_{v\in D^{X}}T_{F,\sigma\rightarrow v}\enskip.\end{eqnarray*}

\end{proof}
We want to insure that transfers made towards a valuation are at least
its unladen weight, hence we define the following property of covering.
\begin{defn}
Let $S$ be a subset of $\mathcal{S}\left(F\right)$; we say that
$\left(T_{F},S\right)$ \emph{covers }$U_{F}$ iff $\forall v\in D^{x},\sum_{\sigma\in S}T_{F,\sigma\to v}\geq U_{F}\left(v\right)$.
\end{defn}
We can now state some general conditions that are sufficient for a
weighting scheme to be correct.

\subsection{Weight Conservation Theorem}
\begin{thm}
\label{Weight-Conservation-Theorem}(Weight Conservation Theorem).
If the following assumptions hold:
\begin{enumerate}
\item the weighting seed $s_{F}$ is unitary,
\item the actual weight $w_{F}$ is decomposable by family $\left(\delta_{F,\sigma,x,a}\right)$,
\item $\left(T_{F},S\right)$ covers $U_{F}$,
\end{enumerate}
then $W_{F}\left(S\right)\geq1$.\end{thm}
\begin{proof}
Since $w_{F}$ is decomposable by family $\left(\delta_{F,\sigma,x,a}\right)$,
lemma \ref{lem:Transfer} asserts that $\forall\sigma\in S,W_{F}\left(\sigma\right)=\sum_{v\in D^{X}}T_{F,\sigma\to v}$.
Thus \begin{eqnarray*}
W_{F}\left(S\right) & = & \sum_{\sigma\in S}W_{F}\left(\sigma\right)\\
 & = & \sum_{\sigma\in S}\sum_{v\in D^{X}}T_{F,\sigma\to v}\mbox{ by lemma \ref{lem:Transfer}}\\
 & = & \sum_{v\in D^{X}}\sum_{\sigma\in S}T_{F,\sigma\to v}\\
 & \geq & \sum_{v\in D^{X}}U_{F}\left(v\right)\mbox{ since \ensuremath{\left(T_{F},S\right)}\,\ covers \ensuremath{U_{F}}}\enskip.\end{eqnarray*}

Moreover by lemma \ref{lem:unitary}, since $s_{F}$ is unitary, $\sum_{v\in D^{X}}U_{F}\left(v\right)=1$.
\end{proof}
Thus we have exhibited three sufficient conditions to get a weight
conservation theorem. These conditions might not be necessary; however
not any weighting system $w_{F}$ will be correct, as shown in example
on figure \ref{fig:bad}. So let us introduce a way to build $w_{F}$
from $s_{F}$ in a way that is intended to match the conditions of
our Weight Conservation Theorem.

\subsection{Generators}

All weights we put onto solutions (either in section \ref{sub:Seeds-and-dispatchers}
or in section \ref{sec:Better-than-Maneva}) are built from a \emph{weight
generator}, as follows.
\begin{defn}
\label{def:generators}A \emph{generator} is a function $\omega_{F}\,:\, X\times D\times\mathcal{P}\left(D\right)\to\boldsymbol{R}^{+}$.
We say that $\omega_{F}$ is \emph{unitary }iff for all variable $x$
and all nonempty subset $\Delta$ of $D$, $\sum_{a\in\Delta}\omega_{F}\left(x,a,\Delta\right)=1$.

From the weight generator $\omega_{F}$ we now define the \emph{actual
weight} $w_{F}$ of a variable in a solution:\begin{eqnarray}
w_{F}\left(\sigma,x\right) & = & \omega_{F}\left(x,\sigma\left(x\right),A_{F}\left(\sigma,x\right)\right)\enskip.\end{eqnarray}
\end{defn}
\begin{rem*}
If $\sigma$ and $\tau$ are 2 solutions such that $\sigma\left(x\right)=\tau\left(x\right)$
and $A_{F}\left(\sigma,x\right)=A_{F}\left(\tau,x\right)$, then $w_{F}\left(\sigma,x\right)=w_{F}\left(\tau,x\right)$.
This is what \citet{Boufkhad2010} call a \emph{uniform} weighting.
\end{rem*}
This may suggest that it could be sufficient to put any weights such
that the sum of weights on any clique would be $1$; but it is not
the case (cf. example on figure \ref{fig:bad}).

\subsection{\label{sub:Seeds-and-dispatchers}Dispatchers}
\begin{defn}
\label{def:dispatchers}A \emph{dispatcher} is a function $d_{F}\,:\, X\times D\to\boldsymbol{R}_{*}^{+}$.

Using the weighting seed $s_{F}$ and the dispatcher $d_{F}$ we now
build the weight generator $\omega_{F}$ of variables in a solution.
Each variable will keep its seed $s_{F}$; moreover the weights of
forbidden values will be dispatched to allowed values thanks to $d_{F}$,
in the following way: \begin{eqnarray}
\omega_{F}\left(x,a,\Delta\right) & = & \begin{cases}
s_{F}\left(x,a\right)+\frac{d_{F}\left(x,a\right)}{\sum_{b\in\Delta}d_{F}\left(x,b\right)}\sum_{b\in D\backslash\Delta}s_{F}\left(x,b\right) & \mbox{if \ensuremath{a\in\Delta}};\\
0 & \mbox{otherwise}\enskip.\end{cases}\label{eq:omegaF}\end{eqnarray}

$\Delta$ represents a category of set of allowed values; so the dispatcher
$d_{F}$ dispatches the total weighting seed of forbidden values among
allowed values.\end{defn}
\begin{fact*}
If $s_{F}$ is unitary, so is $\omega_{F}$.\end{fact*}
\begin{defn}
We say that the weighting system is \emph{homogeneous} when $d_{F}=s_{F}$.
In this noticeable case the same function is used to assign a weighting
seed and to dispatch remaining weights among neighbors.
\end{defn}

\paragraph*{Examples of Weightings}

As one can see in figure \ref{fig:bad}, even if we put a total weight
of $1$ on each clique,  the overall weight can be less than $1$.
To prevent such bad configurations we let our weights take the form
of seeds+dispatchers (figures \ref{fig:homogeneous} and \ref{fig:heterogeneous}).
The purpose of building weights from seeds and dispatchers is to prevent
the same kind of inversions that we encountered for orientations (which
led to circuits): in figure \ref{fig:bad}, in clique $\left\{ a,b\right\} $
for variable $x$, $a$ is given a much smaller weight than $b$,
whereas in clique $\left\{ a,b,c\right\} $ for the same variable
$x$, the opposite occurs. In fact dispatchers allow some reshuffling
of weights between different cliques of the same variable (see the
weights given to values $a$ and $b$ in cliques $\left\{ a,b\right\} $
and $\left\{ a,b,c\right\} $ for variable $x$ on figure \ref{fig:heterogeneous}),
but the fact that seeds and dispatchers are assigned to each individual
couple (variable, value) enables a kind of consistency between a clique
and its sub-cliques, preventing circuit-like structures from appearing.
\begin{itemize}
\item figure \ref{fig:homogeneous} was obtained by the following choice
of $s_{F}$ and $d_{F}$ (homogeneous case, so $d_{F}=s_{F}$):\\
\begin{tabular}{|c|c|c|c|}
\hline 
$s_{F}$ & $a$ & $b$ & $c$\tabularnewline
\hline
\hline 
$x$ & $0.1$ & $0.2$ & $0.7$\tabularnewline
\hline 
$y$ & $0.4$ & $0.3$ & $0.3$\tabularnewline
\hline
\end{tabular}\qquad{}\begin{tabular}{|c|c|c|c|}
\hline 
$d_{F}$ & $a$ & $b$ & $c$\tabularnewline
\hline
\hline 
$x$ & $0.1$ & $0.2$ & $0.7$\tabularnewline
\hline 
$y$ & $0.4$ & $0.3$ & $0.3$\tabularnewline
\hline
\end{tabular}~;
\item figure \ref{fig:heterogeneous} was obtained by the following choice
of $s_{F}$ and $d_{F}$:\\
\begin{tabular}{|c|c|c|c|}
\hline 
$s_{F}$ & $a$ & $b$ & $c$\tabularnewline
\hline
\hline 
$x$ & $0.1$ & $0.2$ & $0.7$\tabularnewline
\hline 
$y$ & $0.4$ & $0.3$ & $0.3$\tabularnewline
\hline
\end{tabular}\qquad{}\begin{tabular}{|c|c|c|c|}
\hline 
$d_{F}$ & $a$ & $b$ & $c$\tabularnewline
\hline
\hline 
$x$ & $0.6$ & $0.3$ & $0.1$\tabularnewline
\hline 
$y$ & $0.2$ & $0.5$ & $0.3$\tabularnewline
\hline
\end{tabular}~.
\end{itemize}
\begin{figure}
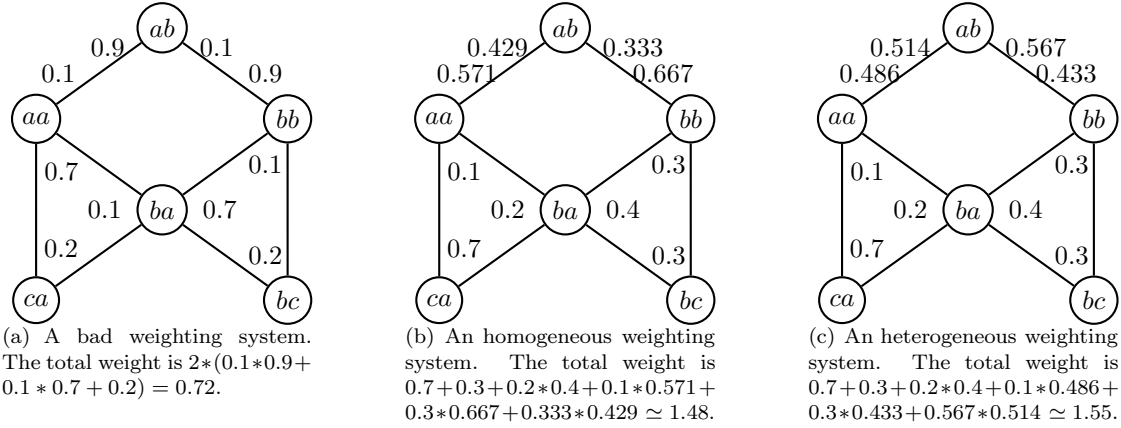

\centering
\subfigure[A bad weighting system. The total weight is $2*(0.1*0.9+0.1*0.7+0.2)=0.72$.]{\label{fig:bad}
$ \psmatrix[colsep=1cm,rowsep=0.5cm,mnode=circle] 
&ab\\  
aa&&bb\\
&ba\\
ca&&bc
\nput[labelsep=0.3]{67.5}{4,1}{0.2}
\nput[labelsep=0.3]{-67.5}{2,1}{0.7} \nput[labelsep=0.2]{67.5}{2,1}{0.1}
\nput[labelsep=0.2]{-157.5}{1,2}{0.9} \nput[labelsep=0.2]{-22.5}{1,2}{0.1}
\nput[labelsep=0.2]{180}{3,2}{0.1} \nput[labelsep=0.2]{0}{3,2}{0.7}
\nput[labelsep=0.2]{112.5}{2,3}{0.9} \nput[labelsep=0.2]{-112.5}{2,3}{0.1}
\nput[labelsep=0.2]{112.5}{4,3}{0.2}
\ncline{-}{2,1}{1,2}
\ncline{-}{1,2}{2,3} 
\ncline{-}{3,2}{2,3}
\ncline{-}{2,1}{3,2}
\ncline{-}{2,1}{4,1}
\ncline{-}{4,1}{3,2}
\ncline{-}{2,3}{4,3}
\ncline{-}{3,2}{4,3}
\endpsmatrix  $}
\hspace{1cm}
\subfigure[An homogeneous weighting system. The total weight is $0.7+0.3+0.2*0.4+0.1*0.571+0.3*0.667+0.333*0.429 \simeq 1.48$.]{\label{fig:homogeneous}
$ \psmatrix[colsep=1cm,rowsep=0.5cm,mnode=circle] 
&ab\\  
aa&&bb\\
&ba\\
ca&&bc
\nput[labelsep=0.3]{67.5}{4,1}{0.7}
\nput[labelsep=0.3]{-67.5}{2,1}{0.1} \nput[labelsep=0.2]{67.5}{2,1}{0.571}
\nput[labelsep=0.2]{-157.5}{1,2}{0.429} \nput[labelsep=0.2]{-22.5}{1,2}{0.333}
\nput[labelsep=0.2]{180}{3,2}{0.2} \nput[labelsep=0.2]{0}{3,2}{0.4}
\nput[labelsep=0.2]{112.5}{2,3}{0.667} \nput[labelsep=0.2]{-112.5}{2,3}{0.3}
\nput[labelsep=0.2]{112.5}{4,3}{0.3}
\ncline{-}{2,1}{1,2}
\ncline{-}{1,2}{2,3} 
\ncline{-}{3,2}{2,3}
\ncline{-}{2,1}{3,2}
\ncline{-}{2,1}{4,1}
\ncline{-}{4,1}{3,2}
\ncline{-}{2,3}{4,3}
\ncline{-}{3,2}{4,3}
\endpsmatrix  $}
\hspace{1cm}
\subfigure[An heterogeneous weighting system. The total weight is $0.7+0.3+0.2*0.4+0.1*0.486+0.3*0.433+0.567*0.514 \simeq 1.55$.]{\label{fig:heterogeneous}
$ \psmatrix[colsep=1cm,rowsep=0.5cm,mnode=circle] 
&ab\\  
aa&&bb\\
&ba\\
ca&&bc
\nput[labelsep=0.3]{67.5}{4,1}{0.7}
\nput[labelsep=0.3]{-67.5}{2,1}{0.1} \nput[labelsep=0.2]{67.5}{2,1}{0.486}
\nput[labelsep=0.2]{-157.5}{1,2}{0.514} \nput[labelsep=0.2]{-22.5}{1,2}{0.567}
\nput[labelsep=0.2]{180}{3,2}{0.2} \nput[labelsep=0.2]{0}{3,2}{0.4}
\nput[labelsep=0.2]{112.5}{2,3}{0.433} \nput[labelsep=0.2]{-112.5}{2,3}{0.3}
\nput[labelsep=0.2]{112.5}{4,3}{0.3}
\ncline{-}{2,1}{1,2}
\ncline{-}{1,2}{2,3} 
\ncline{-}{3,2}{2,3}
\ncline{-}{2,1}{3,2}
\ncline{-}{2,1}{4,1}
\ncline{-}{4,1}{3,2}
\ncline{-}{2,3}{4,3}
\ncline{-}{3,2}{4,3}
\endpsmatrix  $}
\caption{Some basic examples of weights. Notations are the same as in figure \ref{fig:basic-examples-orientations}.}
\label{fig:basic-examples-weights}
\end{figure}

We come back to our weighting system $w_{F}$ built from $s_{F}$
and $d_{F}$ and show that it may be used to estimate satisfiability
if $s_{F}$ is unitary. So our first result concerning this weighting
system states that this system is correct for the estimation of satisfiability
(theorem \ref{thm:dispatchers-are-good} below). To prove it, we use
our Weight Conservation Theorem, using the following decomposers:

\begin{eqnarray}
\delta_{F,\sigma,x,a} & = & \begin{cases}
s_{F}\left(x,a\right) & \mbox{if }\sigma\left(x\right)=a;\\
\frac{d_{F}\left(x,\sigma\left(x\right)\right)}{\sum_{a\in A_{F}\left(\sigma,x\right)}d_{F}\left(x,a\right)}s_{F}\left(x,a\right) & \mbox{if \ensuremath{a\notin A_{F}\left(\sigma,x\right)};}\\
0 & \mbox{otherwise\enskip.}\end{cases}\label{eq:delta-dispatchers}\end{eqnarray}

We must now prove that the conditions of our Weight Conservation Theorem
are satisfied: $w_{F}$ is decomposable family $\left(\delta_{F,\sigma,x,a}\right)$
and $\left(T_{F},g\right)$ covers $U_{F}$.
\begin{lem}
\label{lem:wF-is-decomposable}$w_{F}$ is decomposable by family
$\left(\delta_{F,\sigma,x,a}\right)$.\end{lem}
\begin{proof}
By definitions: \begin{eqnarray*}
\sum_{a\in D}\delta_{F,\sigma,x,a} & = & \sum_{a\in D}\left(s_{F}\left(x,a\right)\boldsymbol{1}_{a=\sigma\left(x\right)}+\frac{d_{F}\left(x,\sigma\left(x\right)\right)}{\sum_{a\in A_{F}\left(\sigma,x\right)}d_{F}\left(x,a\right)}s_{F}\left(x,a\right)\boldsymbol{1}_{a\notin A_{F}\left(\sigma,x\right)}\right)\\
 & = & s_{F}\left(x,\sigma\left(x\right)\right)+\frac{d_{F}\left(x,\sigma\left(x\right)\right)}{\sum_{a\in A_{F}\left(\sigma,x\right)}d_{F}\left(x,a\right)}\sum_{a\notin A_{F}\left(\sigma,x\right)}s_{F}\left(x,a\right)\\
 & = & w_{F}\left(\sigma,x\right)\enskip.\end{eqnarray*}

\end{proof}
As the unladen weight of a valuation is scattered among lots of solutions,
in the proof of the following lemma we use an algorithm building a
tree in order to catch enough solutions to insure the covering condition.
The proof is somewhat technical and may be skipped at first reading.
\begin{lem}
\label{lem:(T_F,g)-covers-U_F}Let $g$ be any connected component
of the solutions network $\mathcal{S}\left(F\right)$. Then $\left(T_{F},g\right)$
covers $U_{F}$.\end{lem}
\begin{proof}
First we need some definitions. A \emph{partial valuation $\eta$
over $Y\subseteq X$} is a function from $Y$ to the set $D$. The
\emph{domain }of $\eta$ is $\mathrm{Dom}\left(\eta\right)=Y$. The
\emph{level }(of undetermination) of $\eta$ is $\mathrm{Level}\left(\eta\right)=\left|X\backslash Y\right|$.
Let $Z\subseteq Y\subseteq X$, let $\iota$ be a partial valuation
over $Z$ and $\eta$ be a partial valuation over $Y$. Since $Z\subseteq Y$,
we say that $\iota\leq_{\mathrm{Dom}}\eta$. Of course $\leq_{\mathrm{Dom}}$
is a partial order relation. We say that $\eta$ is an \emph{extension}
of $\iota$ iff $\forall z\in Z,\eta\left(z\right)=\iota\left(z\right)$,
in which case we also say that $\iota$ is the restriction of $\eta$
to $Z$: $\iota=\eta_{|Z}$. In the particular case when $Y=Z\cup\left\{ x\right\} $
with $x\notin Z$, we denote by $\iota_{x\mapsto a}$, the extension
of $\iota$ to $Y$ assigning value $a$ to $x$. Let $g$ be a connected
component of the solutions network. Note that the empty valuation
$\epsilon$ (with domain $\emptyset$) is extensible to a solution
in $g$ as soon as $g\neq\emptyset$. Given a partial valuation $\eta$,
we call $E_{g}\left(\eta\right)$ the set of its extensions which
are elements of $g$ and $r_{g}\left(\eta\right)$ the set of restrictions
of $\eta$ extensible to a solution in $g$ (i.e. restrictions $r$
of $\eta$ such that $E_{g}\left(r\right)\neq\emptyset$).

Let us take any valuation $v$. We must prove that $\sum_{\sigma\in g}T_{F,\sigma\to v}\geq U_{F}\left(v\right)$.
Since $g\neq\emptyset$, $\epsilon\in r_{g}\left(v\right)$ so $r_{g}\left(v\right)\neq\emptyset$
and we can pick an element $v_{0}$ in $r_{g}\left(v\right)$ maximal
with respect to the order $\leq_{\mathrm{Dom}}$. We arbitrarily put
indices $1\dots n_{0}$ onto the remaining $n_{0}=\mathrm{Level}\left(v_{0}\right)$
variables: $x_{1},\dots,x_{n_{0}}$ (i.e. variables not set by $v_{0}$).
In the following algorithm we shall bind a \emph{fictitious weight
}$f\left(\eta\right)$ and a solution $\tau\left(\eta\right)$ to
a partial valuation $\eta$. At the beginning $f\left(v_{0}\right)=U_{F}\left(v\right)$,
and we make a call of $\mathtt{Extend}\left(v_{0}\right)$.

\begin{algorithm}
\caption{\label{alg:Extend}{Extensions of a partial valuation.}}

\begin{algorithmic}[1]

\Procedure{$\mathtt{Extend}$}{$\eta$}

\State$i\gets\mathtt{Level}\left(\eta\right)$

\If{$i=0$}

\State$S\gets S\cup\left\{ \eta\right\} $

\Else

\State\label{tau}$\tau\left(\eta\right)\gets$ a solution maximizing
$\sum_{b\in A_{F}\left(\sigma,x_{i}\right)}d_{F}\left(x_{i},b\right)$
among $\sigma\in E_{g}\left(\eta\right)$ 

\ForAll{$a\in A_{F}\left(\tau\left(\eta\right),x_{i}\right)$}

\State$f\left(\eta_{x_{i}\mapsto a}\right)\gets\frac{d_{F}\left(x_{i},a\right)}{\sum_{b\in A_{F}\left(\tau\left(\eta\right),x\right)}d_{F}\left(x_{i},b\right)}f\left(\eta\right)$

\State$\mathtt{Extend}\left(\eta_{x_{i}\mapsto a}\right)$

\EndFor

\EndIf

\EndProcedure

\end{algorithmic}
\end{algorithm}

Informally we are building a tree and propagating weights from the
root $v_{0}$ (at level $n_{0}$) to leafs which are solutions (at
level $0$) in a conservative way: the total fictitious weight on
level $i$ will be the same as that of level $i+1$.

Formally, what can we insure along this process?
\begin{enumerate}
\item The first thing to notice is that the algorithm stops; namely the
nested calls of $\mathtt{Extend}\left(\eta\right)$ decrement $\mathrm{Level}\left(\eta\right)$
till it reaches $0$.
\item \label{S-in-g}Secondly $S$ is indeed a set of solutions in $g$
extending $v_{0}$. Namely at each call of $\mathtt{Extend}\left(\eta\right)$,
$\eta$ is extensible to a solution in $g$ and the set of unset variables
of $\eta$ is $\left\{ x_{1},\dots,x_{i}\right\} $, where $i=\mathrm{Level}\left(\eta\right)$.
Thus when $i=0$, $\eta$ is a solution in $g$. We prove this by
induction:

\begin{enumerate}
\item at the beginning: $v_{0}\in r_{g}\left(v\right)$, $v_{0}$ trivially
extends itself, $E_{g}\left(v_{0}\right)\neq\emptyset$ and the set
of unset variables of $v_{0}$ is $\left\{ x_{1},\dots,x_{n_{0}}\right\} $;
\item now suppose that $E_{g}\left(\eta\right)\neq\emptyset$, $\eta$ extends
$v_{0}$ and the unset variables of $\eta$ are $\left\{ x_{1},\dots,x_{i}\right\} $;
given $\tau\left(\eta\right)\in E_{g}\left(\eta\right)$, let $a\in A_{F}\left(\tau\left(\eta\right),x_{i}\right)$;
then the valuation $\tau\left(\eta\right)_{x_{i}\gets a}$ is a solution
by definition of $A_{F}\left(\tau\left(\eta\right),x_{i}\right)$;
moreover it is connected to $\tau\left(\eta\right)$, thus $\tau\left(\eta\right)_{x_{i}\gets a}$
is an element of component $g$. Moreover since $\tau\left(\eta\right)$
is an extension of $\eta$ and $x_{i}$ is unset in $\eta$, $\tau\left(\eta\right)_{x_{i}\gets a}$
is an extension of $\eta_{x_{i}\mapsto a}$. Thus $\tau\left(\eta\right)_{x_{i}\gets a}\in E_{g}\left(\eta_{x_{i}\mapsto a}\right)$,
so $E_{g}\left(\eta_{x_{i}\mapsto a}\right)\neq\emptyset$. Of course,
$\eta_{x_{i}\mapsto a}$ extends $v_{0}$, the unset variables of
$\eta_{x_{i}\mapsto a}$ are $\left\{ x_{1},\dots,x_{i-1}\right\} $
and $\mathrm{Level}\left(\eta_{x_{i}\mapsto a}\right)=\mathrm{Level}\left(\eta\right)-1=i-1$.
\end{enumerate}
\item \label{S-yields-U}$\sum_{\sigma\in S}f\left(\sigma\right)=U_{F}\left(v\right)$;
namely among partial valuations considered in the process, $\eta\in S$
iff $\mathrm{Level}\left(\eta\right)=0$. Moreover we now prove by
induction that $\sum_{\mathrm{Level}\left(\eta\right)=i}f\left(\eta\right)=U_{F}\left(v\right)$:

\begin{enumerate}
\item at the beginning when $i=n_{0}$, the only partial valuation of level
$n_{0}$ is $v_{0}$ and $f\left(v_{0}\right)=U_{F}\left(v\right)$;
\item now suppose that $\sum_{\mathrm{Level}\left(\eta\right)=i}f\left(\eta\right)=U_{F}\left(v\right)$;
in our process each partial valuation $\eta$ of level $i-1$ has
one and only one parent in level $i$, which is given by the restriction
$\eta'$ of $\eta$ to $\mathrm{Dom}\left(v_{0}\right)\cup\left\{ x_{i+1},\dots,x_{n_{0}}\right\} $;
thus \begin{eqnarray*}
\sum_{\mathrm{Level}\left(\eta\right)=i-1}f\left(\eta\right) & = & \sum_{\mathrm{Level}\left(\eta'\right)=i}\sum_{a\in A_{F}\left(\tau\left(\eta'\right),x_{i}\right)}f\left(\eta'_{x_{i}\mapsto a}\right)\\
 & = & \sum_{\mathrm{Level}\left(\eta'\right)=i}\sum_{a\in A_{F}\left(\tau\left(\eta'\right),x_{i}\right)}\frac{d_{F}\left(x_{i},a\right)}{\sum_{b\in A_{F}\left(\tau\left(\eta'\right),x_{i}\right)}d_{F}\left(x_{i},b\right)}f\left(\eta'\right)\\
 & = & \sum_{\mathrm{Level}\left(\eta'\right)=i}f\left(\eta'\right)\\
 & = & U_{F}\left(v\right)\enskip.\end{eqnarray*}

\end{enumerate}
\item \label{equal}$\forall\sigma\in S,\forall i\in\left\{ 1,\dots,n_{0}\right\} ,v\left(x_{i}\right)\notin A_{F}\left(\sigma,x_{i}\right)$.
Suppose on the contrary that $\exists\sigma\in S,\exists i\in\left\{ 1,\dots,n_{0}\right\} ,v\left(x_{i}\right)\in A_{F}\left(\sigma,x_{i}\right)$;
the partial valuation $v_{0\, x_{i}\mapsto v\left(x_{i}\right)}$
is still a restriction of $v$; moreover, since by item \ref{S-in-g},
$\sigma$ is an extension of $v_{0}$, $\sigma_{x_{i}\gets v\left(x_{i}\right)}$
is an extension of $v_{0\, x_{i}\mapsto v\left(x_{i}\right)}$; and
since $v\left(x_{i}\right)\in A_{F}\left(\sigma,x_{i}\right)$, $\sigma_{x_{i}\gets v\left(x_{i}\right)}$
is a solution. Thus $v_{0\, x_{i}\mapsto v\left(x_{i}\right)}\in r_{g}\left(v\right)$
and $v_{0\, x_{i}\mapsto v\left(x_{i}\right)}>_{\mathrm{Dom}}v_{0}$,
contradicting the maximality of $v_{0}$ in $r_{g}\left(v\right)$.
\item \label{lower-Omega}$\forall\sigma\in S,f\left(\sigma\right)\leq T_{F,\sigma\to v}$;
namely, let us take any $\sigma\in S$:\begin{eqnarray*}
T_{F,\sigma\to v} & = & \prod_{x\in X}\delta_{F,\sigma,x,v\left(x\right)}\mbox{ by definition \ref{def:T}}\\
 & = & \prod_{x\in X}s_{F}\left(x,v\left(x\right)\right)\left(\boldsymbol{1}_{v\left(x\right)=\sigma\left(x\right)}+\frac{d_{F}\left(x,\sigma\left(x\right)\right)}{\sum_{a\in A_{F}\left(\sigma,x\right)}d_{F}\left(x,a\right)}\boldsymbol{1}_{v\left(x\right)\notin A_{F}\left(\sigma,x\right)}\right)\mbox{ by eq. \ref{eq:delta-dispatchers}}\\
 & = & \prod_{x\in\mathrm{Dom}\left(v_{0}\right)}s_{F}\left(x,v\left(x\right)\right)\prod_{i=1}^{n_{0}}\frac{d_{F}\left(x_{i},\sigma\left(x_{i}\right)\right)s_{F}\left(x_{i},v\left(x_{i}\right)\right)}{\sum_{a\in A_{F}\left(\sigma,x_{i}\right)}d_{F}\left(x_{i},a\right)}\mbox{ by item \ref{S-in-g} and \ref{equal}}\\
 & = & U_{F}\left(v\right)\prod_{i=1}^{n_{0}}\frac{d_{F}\left(x_{i},\sigma\left(x_{i}\right)\right)}{\sum_{a\in A_{F}\left(\sigma,x_{i}\right)}d_{F}\left(x_{i},a\right)}\mbox{ by definition \ref{def:U}}\enskip.\end{eqnarray*}
 Moreover note that \begin{eqnarray*}
f\left(\sigma\right) & = & f\left(\sigma_{|\mathrm{Dom}\left(v_{0}\right)}\right)\prod_{i=1}^{n_{0}}\frac{f\left(\sigma_{|\mathrm{Dom}\left(v_{0}\right)\cup\left\{ x_{i},\dots,x_{n_{0}}\right\} }\right)}{f\left(\sigma_{|\mathrm{Dom}\left(v_{0}\right)\cup\left\{ x_{i+1},\dots,x_{n_{0}}\right\} }\right)}\\
 & = & f\left(v_{0}\right)\prod_{i=1}^{n_{0}}\frac{f\left(\sigma_{|\mathrm{Dom}\left(v_{0}\right)\cup\left\{ x_{i+1},\dots,x_{n_{0}}\right\} ,x_{i}\mapsto\sigma\left(x_{i}\right)}\right)}{f\left(\sigma_{|\mathrm{Dom}\left(v_{0}\right)\cup\left\{ x_{i+1},\dots,x_{n_{0}}\right\} }\right)}\\
 & = & U_{F}\left(v\right)\prod_{i=1}^{n_{0}}\frac{d_{F}\left(x_{i},\sigma\left(x_{i}\right)\right)}{\sum_{a\in A_{F}\left(\tau\left(\sigma_{|\mathrm{Dom}\left(v_{0}\right)\cup\left\{ x_{i+1},\dots,x_{n_{0}}\right\} }\right),x_{i}\right)}d_{F}\left(x_{i},a\right)}\enskip.\end{eqnarray*}
Since of course, for all $i$ between $1$ and $n_{0}$, $\sigma\in E_{g}\left(\sigma_{|\mathrm{Dom}\left(v_{0}\right)\cup\left\{ x_{i+1},\dots,x_{n_{0}}\right\} }\right)$,
by choice of $\tau\left(\eta\right)$ in line \ref{tau} of algorithm
\ref{alg:Extend}, we have that $\sum_{a\in A_{F}\left(\tau\left(\sigma_{|\mathrm{Dom}\left(v_{0}\right)\cup\left\{ x_{i+1},\dots,x_{n_{0}}\right\} }\right),x_{i}\right)}d_{F}\left(x_{i},a\right)\geq\sum_{a\in A_{F}\left(\sigma,x_{i}\right)}d_{F}\left(x_{i},a\right)$,
whence $f\left(\sigma\right)\leq U_{F}\left(v\right)\prod_{i=1}^{n_{0}}\frac{d_{F}\left(x_{i},\sigma\left(x_{i}\right)\right)}{\sum_{a\in A_{F}\left(\sigma,x_{i}\right)}d_{F}\left(x_{i},a\right)}=T_{F,\sigma\to v}$.
\end{enumerate}
Thus we finally get that \begin{eqnarray*}
\sum_{\sigma\in g}T_{F,\sigma\to v} & \geq & \sum_{\sigma\in g}f\left(\sigma\right)\mbox{ by item \ref{lower-Omega}}\\
 & \geq & \sum_{\sigma\in S}f\left(\sigma\right)\mbox{ because \ensuremath{S\subseteq g}}\enskip.\end{eqnarray*}

Moreover, by item \ref{S-yields-U}, $\sum_{\sigma\in S}f\left(\sigma\right)=U_{F}\left(v\right)$;
thus $\left(T_{F},g\right)$ covers $U_{F}$.
\end{proof}
From lemmas \ref{lem:wF-is-decomposable} and \ref{lem:(T_F,g)-covers-U_F},
we conclude that our weighting system built from seeds and dispatchers
obeys the conditions of the Weight Conservation Theorem.
\begin{thm}
\label{thm:dispatchers-are-good}Let $F$ be a satisfiable CSP instance
and $g$ the solutions in a connected component of the solutions network
of $F$. Weights $w_{F}$ are built from seeds $s_{F}$ and dispatchers
$d_{F}$, as in definition \ref{def:dispatchers}. If the weighting
seed $s_{F}$ is unitary, then $W_{F}\left(g\right)\geq1$.\end{thm}
\begin{rem*}
In this paper we do not address the question of choosing the best
$s_{F}$ and $d_{F}$ for a given instance $F$ or for a given family
of instances, which must be custom-tailored depending on the considered
problem. 
\end{rem*}

\section{\label{sec:Better-than-Maneva}Boolean Case: a Better Upper Bound
on the Existence of Non-Trivial Cores}

\subsection{\label{sub:A-Better-Weighting}A Better Weighting for Partial Valuations}

In order to estimate boolean satisfiability of formulas, \citet{Maneva2008}
use a so called \emph{Weight Preservation Theorem}. Valuations here
are mappings from $X$ to $D=\left\{ 0,1,*\right\} $. The value $*$
is to be interpreted as $0$ \emph{or} $1$. They call a valuation
\emph{valid} iff each clause contains at least one true literal or
two starred literals.  In this section, one has to be aware about
the fact that we define a \emph{boolean solution} as a valid valuation
taking its values in $\left\{ 0,1\right\} $ only! $\mathcal{S}\left(F\right)$
still denotes the set of valid valuations of an instance $F$ (with
values in $\left\{ 0,1,*\right\} $) and $A_{F}\left(\sigma,x\right)$
still refers to neighborhood in $\mathcal{S}\left(F\right)$. Note
that any formula has at least one valid valuation: the one that gives
the value $*$ to every variable (the so-called trivial core), so
the existence of valid valuations does not guarantee the existence
of boolean solutions. Nevertheless, counting weighted valid valuations
can be used to estimate boolean satisfiability.

Maneva and Sinclair choose their weights as follows: each variable
has a weighting seed $s_{0}\left(x\right)$, $s_{*}\left(x\right)$
such that $s_{0}\left(x\right)+s_{*}\left(x\right)=1$, and for all
valid valuation $\sigma$ and all variable $x$ they put the following
weight:\begin{eqnarray}
q_{F}\left(\sigma,x\right) & = & \begin{cases}
s_{*}\left(x\right) & \mbox{if \ensuremath{\sigma\left(x\right)=*}\enskip;}\\
s_{0}\left(x\right) & \mbox{if \ensuremath{\sigma\left(x\right)\neq*}\,\ and \ensuremath{*\in A_{F}\left(\sigma,x\right)}\enskip;}\\
s_{0}\left(x\right)+s_{*}\left(x\right) & \mbox{otherwise\enskip.}\end{cases}\label{eq:qF}\end{eqnarray}

As shown by \citet{Maneva2007}, the sum of the weights of all valid
valuations reachable from any boolean solution is exactly $1$. The
reachability property is defined as the existence of a path from the
boolean solution to the valid valuation where at each step a variable
is given the value $*$ while maintaining the validity property. Since
a given valid valuation may be reachable from lots of different boolean
solutions (but sometimes from no one), counting the weighted partial
valuations hopefully enables to count less than the number of boolean
solutions.

Using these weights, \citet{Maneva2008} count the so called non-trivial
cores $\sigma$; a non-trivial core $\sigma$ is a valid valuation
with a linear number of non-starred and non-starrable variables (i.e.~such
that $A_{F}\left(\sigma,x\right)=\{\sigma\left(x\right)\}$). Many
non-trivial cores are not extensible to solutions; a core is extensible
to a solution when there is a boolean valuation of the starred literals
which is a boolean solution. They manage to count only cores which
are extensible to a boolean solution, and they estimate the satisfiability
of the starred part of the formula by weighting valid assignments
as defined in equation \ref{eq:qF}. In this section we define our
new weights and show that they are correct, and in section \ref{sub:Non-trivial-cores}
we use them to improve on Maneva \& Sinclair's upper bound, from $4.453$
to $4.419$.

Before we give this improvement and show its correctness, we want
to stress an important difference between the weighting of solutions
of general CSPs as defined in the previous sections and the weighting
defined in this section: in the previous sections, an unsatisfiable
formula has always a total weight of $0$ while in the present one,
an unsatisfiable formula (a formula with no boolean solution) will
have a non-zero total weight (provided the weights of the value $*$
are not $0$). This is the price one has to pay to lower the weights
of satisfiable formulas. This fact makes difficult to establish a
general comparison between both methods, because they are highly dependent
on the set of instances that are considered and in particular on the
proportion of unsatisfiable instances among them.

To improve on Maneva et al.'s estimation system, we choose the following
weights: each variable $x$ has a unitary weighting seed $s_{F}\left(x,0\right)$,
$s_{F}\left(x,1\right)$ and $s_{F}\left(x,*\right)$. From this seed
$s_{F}$ we define the weight generator $\omega_{F}$ as follows:

\begin{eqnarray}
\omega_{F}\left(x,a,\Delta\right) & = & \begin{cases}
s_{F}\left(x,a\right) & \mbox{if \ensuremath{a=*}\,\ and \ensuremath{a\in\Delta}}\enskip;\\
s_{F}\left(x,a\right)+\sum_{b\in D\backslash\Delta}s_{F}\left(x,b\right) & \mbox{if \ensuremath{a\neq*}\,\ and \ensuremath{a\in\Delta}}\enskip;\\
0 & \mbox{if \ensuremath{a\notin\Delta}}\enskip.\end{cases}\label{eq:pF<Maneva}\end{eqnarray}

As before in section \ref{sub:Seeds-and-dispatchers}, we define the
actual weight $w_{F}\left(\sigma,x\right)=\omega_{F}\left(x,\sigma\left(x\right),A_{F}\left(\sigma,x\right)\right)$.
\begin{rem}
Noticeable values of $\omega_{F}$:\begin{eqnarray*}
\omega_{F}\left(x,0,\left\{ 0\right\} \right)=\omega_{F}\left(x,1,\left\{ 1\right\} \right) & = & s_{F}\left(x,0\right)+s_{F}\left(x,1\right)+s_{F}\left(x,*\right)=1\enskip;\\
\omega_{F}\left(x,0,\left\{ 0,*\right\} \right)=\omega_{F}\left(x,1,\left\{ 1,*\right\} \right) & = & s_{F}\left(x,0\right)+s_{F}\left(x,1\right)\enskip;\\
\omega_{F}\left(x,a,\left\{ 0,1,*\right\} \right) & = & s_{F}\left(x,a\right)\enskip;\\
\omega_{F}\left(x,*,\Delta\right) & = & s_{F}\left(x,*\right)\mbox{ if \ensuremath{*\in\Delta}}\enskip.\end{eqnarray*}

\end{rem}

\begin{rem}
\label{rem:almost-unitary}$\omega_{F}$ is almost unitary, since
for all nonempty $\Delta\subseteq D$, $\Delta\neq\left\{ *\right\} $,
$\Delta\neq\left\{ 0,1\right\} $,\\
 $\sum_{a\in\Delta}\omega\left(x,a,\Delta\right)=1$; $\left\{ 0,1\right\} $
cannot be a clique in this model of validity, because if both $0$
and $1$ are allowed, so is $*$. However $\left\{ *\right\} $ can
be a clique, and in this case $\omega\left(x,*,\left\{ *\right\} \right)=s_{F}\left(x,*\right)$.
\end{rem}
Our system can be seen as a split of $1-s_{*}\left(x\right)$ into
$s_{F}\left(x,0\right)$ and $s_{F}\left(x,1\right)$ (instead of
just $s_{0}\left(x\right)$ for Maneva) in the case when $\sigma\left(x\right)\neq*$
and $A_{F}\left(\sigma,x\right)=\left\{ 0,1,*\right\} $; thus our
weights are smaller than Maneva's, though we are able to insure that
they are correct.

This system is different from the system seeds+dispatchers, because
here a fixed variable at value $*$ is given a weight of $s_{F}\left(x,*\right)$,
whereas dispatchers would give it a weight of $1$. However we are
able to use our Weight Conservation Theorem, using the following decomposers:
\begin{eqnarray}
\delta_{F,\sigma,x,a} & = & \begin{cases}
s_{F}\left(x,a\right) & \mbox{if }\begin{cases}
\mbox{\ensuremath{\sigma\left(x\right)=a}}\\
\mbox{or (\ensuremath{\sigma\left(x\right)\neq*}\,\ and \ensuremath{a\notin A_{F}\left(\sigma,x\right)})}\end{cases};\\
0 & \mbox{otherwise\enskip.}\end{cases}\label{eq:delta-Maneva}\end{eqnarray}

We must now prove that the conditions of our Weight Conservation Theorem
are satisfied: $w_{F}$ is decomposable family $\left(\delta_{F,\sigma,x,a}\right)$
and $\left(T_{F},g\right)$ covers $U_{F}$.
\begin{lem}
$w_{F}$ is decomposable by family $\left(\delta_{F,\sigma,x,a}\right)$.\end{lem}
\begin{proof}
By definitions:
\begin{enumerate}
\item if $\sigma\left(x\right)=*$: $\sum_{a\in D}\delta_{F,\sigma,x,a}=\sum_{a\in D}s_{F}\left(x,a\right)\boldsymbol{1}_{a=\sigma\left(x\right)}=s_{F}\left(x,\sigma\left(x\right)\right)=w_{F}\left(\sigma,x\right)$;
\item if $\sigma\left(x\right)\neq*$: \begin{eqnarray*}
\sum_{a\in D}\delta_{F,\sigma,x,a} & = & \sum_{a\in D}s_{F}\left(x,a\right)\left(\boldsymbol{1}_{a=\sigma\left(x\right)}+\boldsymbol{1}_{a\notin A_{F}\left(\sigma,x\right)}\right)\\
 & = & s_{F}\left(x,\sigma\left(x\right)\right)+\sum_{b\notin A_{F}\left(\sigma,x\right)}s_{F}\left(x,b\right)\\
 & = & w_{F}\left(\sigma,x\right)\enskip.\end{eqnarray*}

\end{enumerate}
\end{proof}
\begin{lem}
Let $v$ be a valuation and $g$ be any connected component of the
network of valid valuations $\mathcal{S}\left(F\right)$ containing
a boolean solution. Then there exists a valid valuation $\sigma\in g$
such that $U_{F}\left(v\right)=T_{F,\sigma\to v}$.\end{lem}
\begin{proof}
Let us take any $v\in D^{X}$ and a boolean solution $\sigma_{0}\in g$.
At the beginning we put $\sigma=\sigma_{0}$. Consider the following
procedure:
\begin{itemize}
\item If there is a variable $x\in X$ such that $\sigma\left(x\right)\neq v\left(x\right)$
and $\sigma_{x\gets v\left(x\right)}$ remains a valid valuation,
then change $\sigma$ to $\sigma_{x\gets v\left(x\right)}$.
\end{itemize}
We iterate this procedure till there is no variable $x\in X$ such
that $\sigma\left(x\right)\neq v\left(x\right)$ and $\sigma_{x\gets v\left(x\right)}$
remains a valid valuation. This eventually happens because at each
step we make a move towards $v$, and $X$ is finite. So in the end,
each variable in $\sigma$ has either its initial boolean value in
$\sigma_{0}$ or the value given by $v$. In other words, for any
$x\in X$, either $\sigma\left(x\right)=v\left(x\right)$ or ($\sigma\left(x\right)\neq*$
and $v\left(x\right)\notin A_{F}\left(\sigma,x\right)$). Thus by
equation \ref{eq:delta-Maneva}, $\delta_{F,\sigma,x,v\left(x\right)}=s_{F}\left(x,v\left(x\right)\right)$,
which in turn by definitions \ref{def:U} and \ref{def:T} yields
$T_{F,\sigma\to v}=U_{F}\left(v\right)$. Moreover, by construction,
the ending $\sigma$ is also in $g$.\end{proof}
\begin{cor}
Let $g$ be any connected component of the network of valid valuations
$\mathcal{S}\left(F\right)$ containing a boolean solution. Then $\left(T_{F},g\right)$
covers $U_{F}$.
\end{cor}
Thus our weighting system obeys the Weight Conservation Theorem, and
we can conclude that $\gamma\left(F\right)=W_{F}\left(\mathcal{S}\left(F\right)\right)\geq W_{F}\left(g\right)\geq1$
and state the following theorem:
\begin{thm}
\label{thm:Maneva}$w_{F}$ as defined in equation \ref{eq:pF<Maneva}
yields $\gamma\left(F\right)\geq1$ whenever $F$ admits a boolean
solution.\end{thm}
\begin{rem*}
We cannot apply theorem \ref{thm:set-equality} because there is no
closure by renaming; namely $\left\{ 0,1\right\} $ cannot be a clique
(if both $0$ and $1$ are allowed, so is $*$), whereas $\left\{ 0,*\right\} $
and $\left\{ 1,*\right\} $ can. 
\end{rem*}
Note that in the particular case where for all $x\in X$, $s_{F}\left(x,0\right)=0$,
$s_{F}\left(x,1\right)=1$ and $s_{F}\left(x,*\right)=0$, we count
what \citet{Dubois1997} call Negatively Prime Solutions (NPSs). Moreover,
as soon as $s_{F}\left(x,*\right)=0$, this weighting can be seen
as seeds+dispatchers on a boolean domain (so this weighting is homogeneous).

We used the weighting defined in equation \ref{eq:pF<Maneva} to compute
an upper bound of the threshold of random \threesat: taking seeds
independent of $F$ and $x$, we obtained the best estimation when
$s_{F}\left(x,*\right)=0$ (and the corresponding upper bound is $4.643$,
just like with NPSs). We conjecture that even if one takes seeds dependent
on $F$ or $x$, the best choice of $s_{F}\left(x,*\right)$ to estimate
boolean unsatisfiability remains indeed $0$. The reason why we think
so, is that, as described in remark \ref{rem:almost-unitary}, $\omega_{F}$
is almost unitary, except for the clique $\left\{ *\right\} $, in
which case $\omega_{F}\left(x,*,\left\{ *\right\} \right)=s_{F}\left(x,*\right)$.

\subsection{\label{sub:Non-trivial-cores}Application: Non-Existence of Non-Trivial
Cores in \threesat}

We apply here the weighting defined in \ref{sub:A-Better-Weighting}
to improve on the upper bound on the existence of non-trivial cores
in \threesat{} shown by \citet{Maneva2008}. First, we recall the
basic notions defined by \citet{Maneva2008} and we reformulate them
according to our notations. Starting from solutions (which are also
valid partial assignments), the following process is iterated: whenever
there is a starrable variable $x$ (i.e. a variable such that under
the current valuation $\sigma$, we have $\sigma\left(x\right)\ne*$
and $A\left(\sigma,x\right)$ contains $*$), consider for the next
step the valuation obtained from $\sigma$ by assigning $*$ to $x$.
This process stops when no such variables exist, and the resulting
partial assignment is called a \emph{core}. Only the cores $\sigma$
for which the set of non-starred variables is not empty are interesting;
they are said to be non-trivial. Such cores contain an important information
on the geometry of the space of solutions and underlie the so-called
clustering that explains the difficulty of solving such instances.

To study the existence of cores, \citet{Maneva2008} make use of \emph{covers}.
A cover is a partial valuation where every non-starred variable is
non-starrable. Obviously, each core is also a cover but the converse
is not true. Indeed, by construction, for a core there exist always
a way of assigning the starred variables so that the formula is satisfied,
while for covers no such condition is guaranteed. The size of a core
or a cover is the number of non-starred variables.

\citet{Maneva2008} compute their upper bound in two steps:
\begin{enumerate}
\item First, they use the first moment method to upper-bound the probability
of existence of covers of certain sizes. This allows them to discard
the ranges of sizes where covers do not exist, and within these ranges
cores cannot exist either;
\item As for the ranges that are not discarded by the previous step, the
probability that covers can be extended to solutions is upper-bounded
using the first moment method through the weighting system.
\end{enumerate}
Implementing this method as described in \citet{Maneva2008} we obtain
the following improvement.
\begin{thm}
\label{thm:4.419}Random instances of \threesat{} with density greater
than $4.419$ have no non-trivial cores with high probability.
\end{thm}
As we mimic \citet{Maneva2008}'s proof, we defer it to the appendix.
Let us just mention the weights we use : $s_{F}\left(x,*\right)=\rho$,
$s_{F}\left(x,1\right)=s_{F}\left(x,0\right)=\left(1-\rho\right)/2$.
We determined the best values for $\rho$ by numerical simulations,
and it depends on the parameter $a$ (size of cores). By symmetry,
any combination of values given to the $s_{F}\left(x,1\right)$ and
$s_{F}\left(x,0\right)$ summing up to $1-\rho$ will give the same
result. 

One can wonder whether this bound can be improved upon using boolean
solutions and ordering. The answer is no since that case is a particular
case of our weighting system (eq.~\ref{eq:pF<Maneva}): the case
where $s_{F}\left(x,*\right)=0$, $s_{F}\left(x,1\right)=1$ and $s_{F}\left(x,0\right)=0$
or vice-versa. It is easy to answer this question in the boolean case
because any ordering can be seen as a particular weighting (by choosing
the weight $1$ for a value and $0$ for the other one). But in the
general case of larger domain CSPs the answer is not so easy. In the
following section, we give general comparisons between orderings and
weightings in different general cases.

\section{Weighting versus Ordering}

\subsection{Partial Ordering of Solutions \label{sub:Partial-ordering} }

Given a CSP instance $F$, various partial orders $\prec_{F}$ can
be defined on the set of solutions such that for every two adjacent
solutions $\sigma$ and $\tau$ of $F$, we have either $\sigma\prec_{F}\tau$
or $\tau\prec_{F}\sigma$. The aim of the partial order here is to
provide a measure on the solutions network through the number of its
minimal elements. Let $\mathcal{M}_{\prec_{F}}\left(F\right)$ be
the set of minimal solutions of $F$ with respect to the order $\prec_{F}$.

In the solutions network of $F$, a partial order $\prec_{F}$ can
be seen as a circuit-free orientation of the edges of the graph such
that an edge goes from $\tau$ to $\sigma$ iff $\sigma\prec_{F}\tau$;
then minimal elements are vertices with no outgoing edges. In general
one seeks partial orderings that have the least number of minimal
elements; however the choice is limited because orderings must be
chosen according to local criteria only.

\paragraph*{Construction of an Ordering}
\begin{defn}
Given a variable $x\in X$, a total strict order $<_{F,x}$ on $D$
gives an orientation between neighboring solutions: $\sigma\prec_{F,x}\tau$
iff $\sigma$ and $\tau$ are $x$-adjacent and $\sigma\left(x\right)<_{F,x}\tau\left(x\right)$.
Note that $\prec_{F,x}$ is a partial strict order on the set of solutions,
but a total strict order in each clique $N_{F}\left(\sigma,x\right)$.

We can bring all partial orders $\prec_{F,x}$ together on the set
of solutions, as follows: if $\sigma$ and $\tau$ are $x$-adjacent
and different, then $\sigma\prec_{F}\tau$ iff $\sigma\prec_{F,x}\tau$.
This is possible because two different solutions $\sigma$ and $\tau$
cannot be both $x$-adjacent and $y$-adjacent for two different variables
$x$ and $y$. We say that $\prec_{F}$ is the orientation on $\mathcal{S}\left(F\right)$
\emph{induced }by the set $\left\{ \left(x,<_{F,x}\right)\right\} _{x\in X}$
.\end{defn}
\begin{lem}
\label{lem:circuit-free}If $\prec_{F}$ is the orientation on $\mathcal{S}\left(F\right)$
induced by a set $\left\{ \left(x,<_{F,x}\right)\right\} _{x\in X}$,
then $\prec_{F}$ is circuit-free.\end{lem}
\begin{proof}
Suppose on the contrary that there exists a circuit $\sigma_{1}\prec_{F}\dots\prec_{F}\sigma_{l}\prec_{F}\sigma_{1}$
for some $l\geq2$. Let us consider the variable $x$ such that $\sigma_{1}\prec_{F,x}\sigma_{2}$.
For any $i\leq l$, either $\sigma_{i}\left(x\right)=\sigma_{i+1}\left(x\right)$
(if $\sigma_{i}$ and $\sigma_{i+1}$ are not $x$-adjacent) or $\sigma_{i}\left(x\right)<_{F,x}\sigma_{i+1}\left(x\right)$
(if $\sigma_{i}$ and $\sigma_{i+1}$ are $x$-adjacent). Thus $\sigma_{1}\left(x\right)<_{F,x}\sigma_{2}\left(x\right)$
and $\sigma_{2}\left(x\right)\leq_{F,x}\sigma_{3}\left(x\right)\leq_{F,x}\dots\leq_{F,x}\sigma_{l}\left(x\right)\leq_{F,x}\sigma_{1}\left(x\right)$:
a contradiction.\end{proof}
\begin{cor}
The transitive closure of $\prec_{F}$ is a strict order relation.
\end{cor}

\paragraph*{Instance Dependent or not}
\begin{itemize}
\item \emph{Instance dependent ordering.} In this case, we put for each
variable $x\in X$ and each CSP instance $F$ a total order $<_{F,x}$
onto the domain $D$ of possible values. As mentioned above, we (partially)
order solutions as follows: let $\sigma\in\mathcal{S}\left(F\right)$
and $\tau\in N_{F}\left(\sigma,x\right)$; we have $\sigma\prec_{F}\tau$
if and only if $\sigma\left(x\right)<_{F,x}\tau\left(x\right)$. The
motivation for the instance dependent ordering is that some syntactic
properties of the CSP instance $F$ can be exploited to define a suitable
order for that instance.
\item \emph{Instance independent ordering}. This is a particular case of
the above ordering, when the total order $<_{x}$ on $D$ does not
depend on $F$. For some problems, no preferred order can be defined
given some instance. This happens in particular when values are indistinguishable
because of the symmetry of the problem (e.g. colors in graph coloring).

\end{itemize}

\paragraph*{Examples of Orientations}

We first give an example (figure \ref{fig:bad-orientation}) of an
orientation which is not circuit-free, even though it was built from
the following local orderings on each individual clique:
\begin{itemize}
\item in cliques $\left\{ a,b\right\} $ for variables $x$ and $y$, we
have $b<a$;
\item in cliques $\left\{ a,b,c\right\} $ for variables $x$ and $y$,
we have $a<c<b$.
\end{itemize}
The problem comes from the fact that $a$ and $b$ are ordered differently
in clique $\left\{ a,b,c\right\} $ and its sub-clique $\left\{ a,b\right\} $,
which led us to consider only orientations built in the following
way: we choose for each variable $x$ a total order $<_{x}$ on the
domain $D$ and use it for each sub-clique of $D$. This is what \citet{Boufkhad2010}
call a \emph{uniform} orientation. Example in figure \ref{fig:good-orientation}
was obtained by the following orders: $c<_{x}b<_{x}a$ and $c<_{y}a<_{y}b$.
This orientation is circuit-free and has two minimal elements. Now
among good orientations, the less minimal elements they have, the
better they are; figure \ref{fig:very-good-orientation}, which was
obtained by the following orders: $c<_{x}b<_{x}a$ and $a<_{y}c<_{y}b$,
gives an example of an orientation with just one minimal element.

\begin{figure}
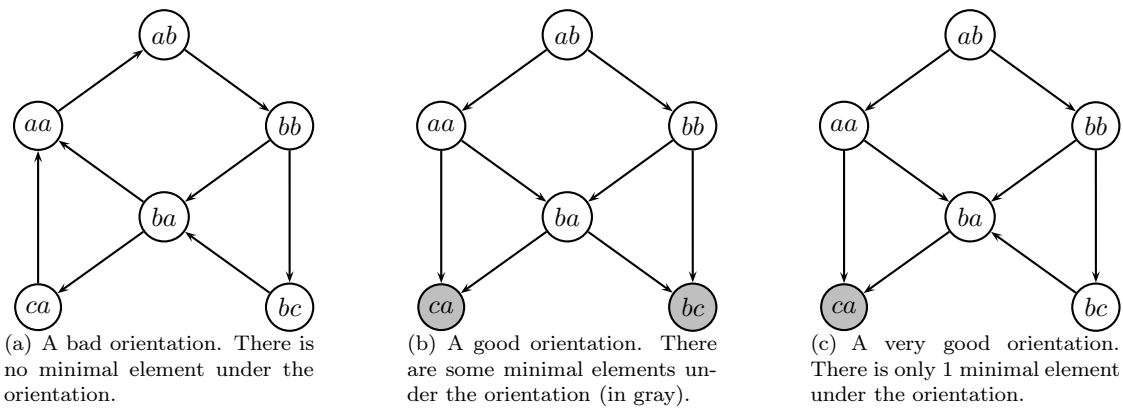

\centering
\subfigure[A bad orientation. There is no minimal element under the orientation.]{\label{fig:bad-orientation}
$ \psmatrix[colsep=1cm,rowsep=0.5cm,mnode=circle] 
&ab\\  
aa&&bb\\
&ba\\
ca&&bc
\ncline{->}{2,1}{1,2}
\ncline{->}{1,2}{2,3} 
\ncline{<-}{3,2}{2,3}
\ncline{<-}{2,1}{3,2}
\ncline{<-}{2,1}{4,1}
\ncline{<-}{4,1}{3,2}
\ncline{->}{2,3}{4,3}
\ncline{<-}{3,2}{4,3}
\endpsmatrix  $}
\hspace{1cm}
\subfigure[A good orientation. There are some minimal elements under the orientation (in gray).]{\label{fig:good-orientation}
$ \psmatrix[colsep=1cm,rowsep=0.5cm,mnode=circle] 
&ab\\  
aa&&bb\\
&ba\\
[fillstyle=solid,fillcolor=gray!50]ca&&[fillstyle=solid,fillcolor=gray!50]bc
\ncline{<-}{2,1}{1,2}
\ncline{->}{1,2}{2,3} 
\ncline{<-}{3,2}{2,3}
\ncline{->}{2,1}{3,2}
\ncline{->}{2,1}{4,1}
\ncline{<-}{4,1}{3,2}
\ncline{->}{2,3}{4,3}
\ncline{->}{3,2}{4,3}
\endpsmatrix  $}
\hspace{1cm}
\subfigure[A very good orientation. There is only 1 minimal element under the orientation.]{\label{fig:very-good-orientation}
$\begin{psmatrix}[colsep=1cm,rowsep=0.5cm,mnode=circle]
&ab\\  
aa&&bb\\
&ba\\
[fillstyle=solid,fillcolor=gray!50]ca&&bc
\ncline{<-}{2,1}{1,2}
\ncline{->}{1,2}{2,3} 
\ncline{<-}{3,2}{2,3}
\ncline{->}{2,1}{3,2}
\ncline{->}{2,1}{4,1}
\ncline{<-}{4,1}{3,2}
\ncline{->}{2,3}{4,3}
\ncline{<-}{3,2}{4,3}
\end{psmatrix}$}
\caption{Some basic examples of orientations. We consider a network of 6 solutions over the domain $D=\{a,b,c\}$ for variables $X=\{x,y\}$; by shortcut $ab$ we mean that variable $x$ takes value $a$ and variable $y$ takes value $b$. $\{aa,ba,ca\}$ is a clique for variable $x$.}
\label{fig:basic-examples-orientations}
\end{figure}

\subsection{Homogeneous Case: Weighting Is Not Better than Ordering}

As we have seen, the weighting is based upon two functions:
\begin{enumerate}
\item the weighting seed $s_{F}$ that determines the intrinsic weight of
each value and then allows to compute the intrinsic unladen weight
of each valuation;
\item the dispatcher $d_{F}$ that represents how the weights of forbidden
valuations are scattered among the authorized ones.
\end{enumerate}
A natural case to investigate is when these two quantities are equal,
namely when each allowed value is dispatched a complimentary weight
proportional to its intrinsic weight. So we deal here with the homogeneous
case $d_{F}=s_{F}$ and show that whatever $s_{F}$ may be, there
will exist an ordering which is at least as good as the weighting
system, as will be stated in theorem \ref{thm:orientation-better}.
The proof consists in choosing variable per variable the order $<_{F,x}$
in a way that does not increase the global weight. For our recurrence
to work we use the homogeneity property. Just as we defined a generator
$\omega_{F}$ for a weight $w_{F}$, so need we now to define a generator
$\mu_{F}$ for an orientation $m_{F}$.
\begin{defn}
We define the following binary weight function: \begin{eqnarray}
\mu_{F}\left(x,a,\Delta\right) & = & \begin{cases}
1 & \mbox{if \ensuremath{a}\,\ is the minimum of \ensuremath{\Delta}\,\ for \ensuremath{<_{F,x}}};\\
0 & \mbox{otherwise}\enskip.\end{cases}\end{eqnarray}
\begin{eqnarray}
m_{F}\left(\sigma,x\right) & = & \mu_{F}\left(x,\sigma\left(x\right),A_{F}\left(\sigma,x\right)\right)\enskip.\end{eqnarray}

\end{defn}
At each step of the recurrence, some variables are ordered while the
other ones are weighted. That leads us to introduce the following
definitions. We are going to substitute binary weights $m_{F}$'s
to original weights $w_{F}$'s variable per variable, so we call $\Xi$
the set of couples of (variables $x$, orders $<_{F,x}$) where $m_{F}$'s
are used and we define
\begin{defn}
\begin{eqnarray}
\Omega_{F}\left(\sigma,\Xi\right) & = & \prod_{x\in\Xi}m_{F}\left(\sigma,x\right)\prod_{x\in X\backslash\Xi}w_{F}\left(\sigma,x\right)\end{eqnarray}

and we extend it to a set $S$ of solutions by\begin{eqnarray}
\Omega_{F}\left(S,\Xi\right) & = & \sum_{\sigma\in S}\Omega_{F}\left(\sigma,\Xi\right)\enskip.\end{eqnarray}
\end{defn}
\begin{rem}
\label{rem:empty-Xsi} What happens when $\Xi$ is empty? \begin{eqnarray}
\Omega_{F}\left(\mathcal{S}\left(F\right),\emptyset\right) & = & W_{F}\left(\mathcal{S}\left(F\right)\right)\enskip.\end{eqnarray}
Namely, by definition, for any solution $\sigma\in\mathcal{S}\left(F\right)$,
$\Omega_{F}\left(\sigma,\emptyset\right)=\prod_{x\in X}w_{F}\left(\sigma,x\right)=W_{F}\left(\sigma\right)$.
But $\Omega_{F}\left(\mathcal{S}\left(F\right),\emptyset\right)=\sum_{\sigma\in\mathcal{S}\left(F\right)}\Omega_{F}\left(\sigma,\emptyset\right)$
and $W_{F}\left(\mathcal{S}\left(F\right)\right)=\sum_{\sigma\in\mathcal{S}\left(F\right)}W_{F}\left(\sigma\right)$.
\end{rem}

\begin{rem}
\label{rem:full-Xsi} What happens when $\Xi$ is full? Suppose that
for all variable $x$, $<_{F,x}$ is a total order on $D$. Let $\prec_{F}$
be the orientation induced by $\left\{ \left(x,<_{F,x}\right)\right\} _{x\in X}$.
Then \begin{eqnarray}
\Omega_{F}\left(\mathcal{S}\left(F\right),\left\{ \left(x,<_{F,x}\right)\right\} _{x\in X}\right) & = & \left|\mathcal{M}_{\prec_{F}}\left(F\right)\right|\enskip.\end{eqnarray}
Namely, let us recall that for any solution $\sigma\in\mathcal{S}\left(F\right)$,
$\Omega_{F}\left(\sigma,\left\{ \left(x,<_{F,x}\right)\right\} _{x\in X}\right)=\prod_{x\in X}m_{F}\left(\sigma,x\right)$.
Thus $\Omega_{F}\left(\sigma,\left\{ \left(x,<_{F,x}\right)\right\} _{x\in X}\right)=1$
iff $\forall x\in X$, $\sigma$ is the minimum of $N_{F}\left(\sigma,x\right)$
for $<_{F,x}$ (or equivalently for $\prec_{F}$); in other words
$\sigma$ is minimal among all of its neighbors, which means that
$\sigma$ is minimal (since $\prec_{F}$ compares neighboring solutions
only). Thus $\Omega_{F}\left(\left\{ \left(x,<_{F,x}\right)\right\} _{x\in X}\right)$
is the number of minimal elements of the underlying orientation $\prec_{F}$.
\end{rem}
We are now ready to state the main lemma in this section.
\begin{lem}
\label{lem:step-x0}Suppose that $s_{F}$ is unitary and $d_{F}=s_{F}$.
Then for each set $\Xi,$  each variable $x_{0}\notin\Xi$, there
exists a total order $<_{F,x_{0}}$ on $D$ such that $\Omega_{F}\left(\mathcal{S}\left(F\right),\Xi\cup\left\{ \left(x_{0},<_{F,x_{0}}\right)\right\} \right)\leq\Omega_{F}\left(\mathcal{S}\left(F\right),\Xi\right)$.
\end{lem}
At first reading it might be convenient to jump directly to theorem
\ref{thm:orientation-better}, because the proof of lemma \ref{lem:step-x0}
is somewhat technical and requires some more notations and sub-lemmas.
We fix a variable $x_{0}\notin\Xi$. Let $a$ be an element of $D$
and $\Delta$ be a subset of $D$. We consider the preimages of $\left(a,\Delta\right)$
obtained through mapping a solution $\sigma$ of instance $F$ to
$\left(\sigma\left(x_{0}\right),A_{F}\left(\sigma,x_{0}\right)\right)$.
We denote these preimages as follows: \begin{eqnarray}
\Sigma_{F,x_{0}}\left(a,\Delta\right) & = & \left\{ \sigma\in\mathcal{S}\left(F\right),\sigma\left(x_{0}\right)=a\mbox{ and }A_{F}\left(\sigma,x_{0}\right)=\Delta\right\} \enskip.\end{eqnarray}
Note that:
\begin{enumerate}
\item when $a\notin\Delta$, $\Sigma_{F,x_{0}}\left(a,\Delta\right)=\emptyset$;
\item \label{enu:partition-adelta}the $\Sigma_{F,x_{0}}\left(a,\Delta\right)$
are pairwise disjoint and $\bigsqcup_{\substack{\Delta\subseteq D\\
a\in\Delta}
}\Sigma_{F,x_{0}}\left(a,\Delta\right)=\mathcal{S}\left(F\right)$;
\item \label{enu:uniform-Sigma}if $\sigma,\tau\in\Sigma_{F,x_{0}}\left(a,\Delta\right)$,
then $w_{F}\left(\sigma,x_{0}\right)=w_{F}\left(\tau,x_{0}\right)=\omega_{F}\left(x_{0},a,\Delta\right)$\\
and $m_{F}\left(\sigma,x_{0}\right)=m_{F}\left(\tau,x_{0}\right)=\mu_{F}\left(x_{0},a,\Delta\right)$;
\item \label{enu:Z}we call \begin{eqnarray}
Z_{F,\Xi,x_{0}}\left(a,\Delta\right) & = & \sum_{\sigma\in\Sigma_{F,x_{0}}\left(a,\Delta\right)}\prod_{x\in\Xi}m_{F}\left(\sigma,x\right)\prod_{x\in X\backslash\left(\Xi\cup\left\{ x_{0}\right\} \right)}w_{F}\left(\sigma,x\right)\enskip;\end{eqnarray}
then by item \ref{enu:uniform-Sigma}, \begin{eqnarray}
\Omega_{F}\left(\Sigma_{F,x_{0}}\left(a,\Delta\right),\Xi\right) & = & \omega_{F}\left(x_{0},a,\Delta\right)\cdot Z_{F,\Xi,x_{0}}\left(a,\Delta\right)\enskip;\label{eq:omegaZ}\\
\Omega_{F}\left(\Sigma_{F,x_{0}}\left(a,\Delta\right),\Xi\cup\left\{ \left(x_{0},<_{F,x_{0}}\right)\right\} \right) & = & \mu_{F}\left(x_{0},a,\Delta\right)\cdot Z_{F,\Xi,x_{0}}\left(a,\Delta\right)\enskip.\label{eq:muZ}\end{eqnarray}

\end{enumerate}
We now need to explore further both quantities we want to compare.
It will be convenient to use the following quantities: let $E\subseteq D$
and $a\in E$; we define the following quantities:\begin{eqnarray}
\zeta_{F,\Xi,x_{0}}\left(a,E\right) & = & \sum_{\substack{\Delta\subseteq E\\
\Delta\ni a}
}Z_{F,\Xi,x_{0}}\left(a,\Delta\right)\enskip;\\
\xi_{F,\Xi,x_{0}}\left(E\right) & = & \sum_{\Delta\subseteq E}\sum_{a\in\Delta}\omega_{F}\left(x_{0},a,\Delta\right)\cdot Z_{F,\Xi,x_{0}}\left(a,\Delta\right)\enskip.\end{eqnarray}

So what is the purpose of introducing these extra quantities? They
will help us prove lemma \ref{lem:step-x0} through the following
facts.
\begin{fact}
\label{lem:Xsi+x0}If $a_{1}<_{F,x_{0}}a_{2}<_{F,x_{0}}\dots<_{F,x_{0}}a_{d}$,
then \begin{eqnarray*}
\Omega_{F}\left(\mathcal{S}\left(F\right),\Xi\cup\left\{ \left(x_{0},<_{F,x_{0}}\right)\right\} \right) & = & \sum_{i=1}^{d}\zeta_{F,\Xi,x_{0}}\left(a_{i},D\backslash\left\{ a_{1},\dots,a_{i-1}\right\} \right)\enskip.\end{eqnarray*}
\end{fact}
\begin{proof}
We use the partition mentioned in item \ref{enu:partition-adelta}:

\begin{eqnarray*}
\Omega_{F}\left(\mathcal{S}\left(F\right),\Xi\cup\left\{ \left(x_{0},<_{F,x_{0}}\right)\right\} \right) & = & \Omega_{F}\left(\bigsqcup_{\substack{\Delta\subseteq D\\
a\in\Delta}
}\Sigma_{F,x_{0}}\left(a,\Delta\right),\Xi\cup\left\{ \left(x_{0},<_{F,x_{0}}\right)\right\} \right)\\
 & = & \sum_{\substack{\Delta\subseteq D\\
a\in\Delta}
}\Omega_{F}\left(\Sigma_{F,x_{0}}\left(a,\Delta\right),\Xi\cup\left\{ \left(x_{0},<_{F,x_{0}}\right)\right\} \right)\\
 & = & \sum_{\substack{\Delta\subseteq D\\
a\in\Delta}
}\mu_{F}\left(x_{0},a,\Delta\right)\cdot Z_{F,\Xi,x_{0}}\left(a,\Delta\right)\mbox{ by eq. \ref{eq:muZ}}\\
 & = & \sum_{a\in D}\sum_{\substack{\Delta\subseteq D\\
\Delta\ni a}
}\mu_{F}\left(x_{0},a,\Delta\right)\cdot Z_{F,\Xi,x_{0}}\left(a,\Delta\right)\\
 & = & \sum_{a\in D}\sum_{\substack{\Delta\subseteq D\\
\Delta\ni a}
}\boldsymbol{1}_{\mbox{\ensuremath{a}\,\ is the minimum of \ensuremath{\Delta}\,\ for \ensuremath{<_{F,x_{0}}}}}\cdot Z_{F,\Xi,x_{0}}\left(a,\Delta\right)\\
 & = & \sum_{i=1}^{d}\sum_{\substack{\Delta\subseteq D\backslash\left\{ a_{1},\dots,a_{i-1}\right\} \\
\Delta\ni a_{i}}
}Z_{F,\Xi,x_{0}}\left(a_{i},\Delta\right)\mbox{ since \ensuremath{a_{1}<_{F,x_{0}}\dots<_{F,x_{0}}a_{d}}}\\
 & = & \sum_{i=1}^{d}\zeta_{F,\Xi,x_{0}}\left(a_{i},D\backslash\left\{ a_{1},\dots,a_{i-1}\right\} \right)\enskip.\end{eqnarray*}
\end{proof}
\begin{fact}
\label{lem:Xsi}For all $x_{0}\notin\Xi$, $\Omega_{F}\left(\mathcal{S}\left(F\right),\Xi\right)=\xi_{F,\Xi,x_{0}}\left(D\right)$.\end{fact}
\begin{proof}
We use again the partition mentioned in item \ref{enu:partition-adelta}:
\begin{eqnarray*}
\Omega_{F}\left(\mathcal{S}\left(F\right),\Xi\right) & = & \Omega_{F}\left(\bigsqcup_{\substack{\Delta\subseteq D\\
a\in\Delta}
}\Sigma_{F,x_{0}}\left(a,\Delta\right),\Xi\right)\\
 & = & \sum_{\substack{\Delta\subseteq D\\
a\in\Delta}
}\Omega_{F}\left(\Sigma_{F,x_{0}}\left(a,\Delta\right),\Xi\right)\\
 & = & \sum_{\substack{\Delta\subseteq D\\
a\in\Delta}
}\omega_{F}\left(x_{0},a,\Delta\right)\cdot Z_{F,\Xi,x_{0}}\left(a,\Delta\right)\mbox{ by eq. \ref{eq:omegaZ}}\\
 & = & \xi_{F,\Xi,x_{0}}\left(D\right)\enskip.\end{eqnarray*}
\end{proof}
\begin{fact}
\label{lem:recurrence-piF}If $E\subseteq D$, $\Delta\subseteq E$,
$a\in\Delta$, $s_{F}$ is unitary and $d_{F}=s_{F}$ then\begin{eqnarray*}
\omega_{F}\left(x,a,\Delta\right)\sum_{b\in E}d_{F}\left(x,b\right) & = & d_{F}\left(x,a\right)+\omega_{F}\left(x,a,\Delta\right)\sum_{b\in E\backslash\Delta}d_{F}\left(x,b\right)\enskip.\end{eqnarray*}
\end{fact}
\begin{proof}
If $a\in\Delta$, then by equation \ref{eq:omegaF}, $\omega_{F}\left(x,a,\Delta\right)\sum_{b\in\Delta}d_{F}\left(x,b\right)=s_{F}\left(x,a\right)\sum_{b\in\Delta}d_{F}\left(x,b\right)+d_{F}\left(x,a\right)\sum_{b\in D\backslash\Delta}s_{F}\left(x,b\right)$.
By equality $d_{F}=s_{F}$ and the fact that $s_{F}$ is unitary,
we get $d_{F}\left(x,a\right)=\omega_{F}\left(x,a,\Delta\right)\sum_{b\in\Delta}d_{F}\left(x,b\right)$.\end{proof}
\begin{fact}
\label{lem:removal}Let $x_{0}\notin\Xi$ and $E$ any nonempty subset
of $D$. Suppose that $s_{F}$ is unitary and $d_{F}=s_{F}$. Then
there exists $a\in E$ such that $\xi_{F,\Xi,x_{0}}\left(E\right)\geq\zeta_{F,\Xi,x_{0}}\left(a,E\right)+\xi_{F,\Xi,x_{0}}\left(E\backslash\left\{ a\right\} \right)$.\end{fact}
\begin{proof}
Let us call $a_{0}$ an element of $E$ minimizing $\zeta_{F,\Xi,x_{0}}\left(a,E\right)+\xi_{F,\Xi,x_{0}}\left(E\backslash\left\{ a\right\} \right)$
when $a\in E$:

\begin{eqnarray*}
\xi_{F,\Xi,x_{0}}\left(E\right)\sum_{b\in E}d_{F}\left(x_{0},b\right) & = & \sum_{b\in E}d_{F}\left(x_{0},b\right)\sum_{\substack{\Delta\subseteq E\\
a\in\Delta}
}\omega_{F}\left(x_{0},a,\Delta\right)Z_{F,\Xi,x_{0}}\left(a,\Delta\right)\\
 & = & \sum_{\substack{\Delta\subseteq E\\
a\in\Delta}
}d_{F}\left(x_{0},a\right)Z_{F,\Xi,x_{0}}\left(a,\Delta\right)\\
 &  & +\sum_{\substack{\Delta\subseteq E\\
a\in\Delta}
}\omega_{F}\left(x_{0},a,\Delta\right)\sum_{b\in E\backslash\Delta}d_{F}\left(x_{0},b\right)Z_{F,\Xi,x_{0}}\left(a,\Delta\right)\mbox{ by fact \ref{lem:recurrence-piF}}\\
 & = & \sum_{\substack{\Delta\subseteq E\\
a\in\Delta}
}d_{F}\left(x_{0},a\right)Z_{F,\Xi,x_{0}}\left(a,\Delta\right)\\
 &  & +\sum_{\substack{\Delta\subseteq E\\
a\in E\backslash\Delta\\
b\in\Delta}
}\omega_{F}\left(x_{0},b,\Delta\right)d_{F}\left(x_{0},a\right)Z_{F,\Xi,x_{0}}\left(b,\Delta\right)\\
 & = & \sum_{a\in E}d_{F}\left(x_{0},a\right)\sum_{\substack{\Delta\subseteq E\\
\Delta\ni a}
}Z_{F,\Xi,x_{0}}\left(a,\Delta\right)\\
 &  & +\sum_{a\in E}d_{F}\left(x_{0},a\right)\sum_{\Delta\subseteq E\backslash\left\{ a\right\} }\sum_{b\in\Delta}\omega_{F}\left(x_{0},b,\Delta\right)Z_{F,\Xi,x_{0}}\left(b,\Delta\right)\\
 & = & \sum_{a\in E}d_{F}\left(x_{0},a\right)\left(\zeta_{F,\Xi,x_{0}}\left(a,E\right)+\xi_{F,\Xi,x_{0}}\left(E\backslash\left\{ a\right\} \right)\right)\\
 & \geq & \sum_{a\in E}d_{F}\left(x_{0},a\right)\left(\zeta_{F,\Xi,x_{0}}\left(a_{0},E\right)+\xi_{F,\Xi,x_{0}}\left(E\backslash\left\{ a_{0}\right\} \right)\right)\enskip.\end{eqnarray*}
That gives what we want since $\sum_{b\in E}d_{F}\left(x_{0},b\right)\neq0$
(by definition \ref{def:dispatchers}, dispatchers must be positive).
 
\end{proof}

\paragraph*{Proof of Lemma \ref{lem:step-x0}}

By fact \ref{lem:Xsi}, $\Omega_{F}\left(\mathcal{S}\left(F\right),\Xi\right)=\xi_{F,\Xi,x_{0}}\left(D\right)$.
From $D$ we successively remove what we call $a_{1},a_{2},\dots,a_{d}$
till we reach the empty set; applying at each step fact \ref{lem:removal}
yields that $\xi_{F,\Xi,x_{0}}\left(D\right)\geq\sum_{i=1}^{d}\zeta_{F,\Xi,x_{0}}\left(a_{i},D\backslash\left\{ a_{1},\dots,a_{i-1}\right\} \right)+\xi_{F,\Xi,x_{0}}\left(\emptyset\right)$.
By definition, $\xi_{F,\Xi,x_{0}}\left(\emptyset\right)=0$. What
order $<_{F,x_{0}}$ shall we choose on $D$?

Of course: $a_{1}<_{F,x_{0}}a_{2}<_{F,x_{0}}\dots<_{F,x_{0}}a_{d}$.

Then by fact \ref{lem:Xsi+x0}, $\Omega_{F}\left(\mathcal{S}\left(F\right),\Xi\cup\left\{ \left(x_{0},<_{F,x_{0}}\right)\right\} \right)=\sum_{i=1}^{d}\zeta_{F,\Xi,x_{0}}\left(a_{i},D\backslash\left\{ a_{1},\dots,a_{i-1}\right\} \right)$.
So in the end $\Omega_{F}\left(\mathcal{S}\left(F\right),\Xi\right)\geq\Omega_{F}\left(\mathcal{S}\left(F\right),\Xi\cup\left\{ \left(x_{0},<_{F,x_{0}}\right)\right\} \right)$.
\begin{thm}
\label{thm:orientation-better}For any instance $F$, any positive
and unitary weighting seed $s_{F}$, when $d_{F}=s_{F}$, there exists
an instance dependent orientation $\prec_{F}$ induced by a set $\left\{ \left(x,<_{F,x}\right)\right\} _{x\in X}$
of total orders on $D$, such that $\left|\mathcal{M}_{\prec_{F}}\left(F\right)\right|\leq W_{F}\left(\mathcal{S}\left(F\right)\right)$.\end{thm}
\begin{proof}
By remark \ref{rem:empty-Xsi}, $W_{F}\left(\mathcal{S}\left(F\right)\right)=\Omega_{F}\left(\mathcal{S}\left(F\right),\emptyset\right)$.
Starting with $\Xi=\emptyset$, we add elements $\left(x_{0},<_{F,x_{0}}\right)$
to $\Xi$ such that $\Omega_{F}\left(\mathcal{S}\left(F\right),\Xi\right)\geq\Omega_{F}\left(\mathcal{S}\left(F\right),\Xi\cup\left\{ \left(x_{0},<_{F,x_{0}}\right)\right\} \right)$,
which is possible by lemma \ref{lem:step-x0}. At the end of the process
we have thus $\Omega_{F}\left(\mathcal{S}\left(F\right),\emptyset\right)\geq\Omega_{F}\left(\mathcal{S}\left(F\right),\left\{ \left(x,<_{F,x}\right)\right\} _{x\in X}\right)$.
Let $\prec_{F}$ be the orientation on $\mathcal{S}\left(F\right)$
induced by $\left\{ \left(x,<_{F,x}\right)\right\} _{x\in X}$.\\
By remark \ref{rem:full-Xsi}, $\Omega_{F}\left(\mathcal{S}\left(F\right),\left\{ \left(x,<_{F,x}\right)\right\} _{x\in X}\right)=\left|\mathcal{M}_{\prec_{F}}\left(F\right)\right|$.
So $W_{F}\left(\mathcal{S}\left(F\right)\right)\geq\left|\mathcal{M}_{\prec_{F}}\left(F\right)\right|$.
\end{proof}
 Whether this theorem is true for heterogeneous weights remains an
open question.
\begin{rem*}
In the particular case of boolean satisfiability (i.e. when $D=\left\{ 0,1\right\} $),
there is no choice on $d_{F}$: the weighting system is necessarily
homogeneous. Thus in this case weighting is not better than ordering.
\end{rem*}

\subsection{Instance Independent Case: Ordering and Weighting Are Equivalent}
\begin{defn}
The weight of a CSP instance $F$ is: \begin{equation}
\gamma\left(F\right)=W_{F}\left(\mathcal{S}\left(F\right)\right)\enskip.\end{equation}
By extension, the weight of a set $\mathcal{F}$ of CSP instances
is: \begin{equation}
\gamma\left(\mathcal{F}\right)=\sum_{F\in\mathcal{F}}\gamma\left(F\right)\enskip.\end{equation}

\end{defn}
A \emph{permutation} over the domain of values is a bijection $\pi:D\rightarrow D$.
A \emph{renaming} of values is a family of permutations $\Pi=\left(\pi_{x}\right)_{x\in X}$
over the domain $D$. For a CSP instance $F$, let $\Pi\left(F\right)$
be the instance where every occurrence of a value $a$ for every variable
$x$ are replaced by $\pi_{x}\left(a\right)$. A set of CSP instances
$\mathcal{F}$ is said to be \emph{closed under renaming} if for any
renaming $\Pi$, if $F\in\mathcal{F}$ then $\Pi\left(F\right)\in\mathcal{F}$.
By abuse of notation, for any valuation $v$, we denote by $\Pi\left(v\right)$
the valuation that assigns value $\pi_{x}\left(v\left(x\right)\right)$
to variable $x$.

Let us first give a very simple yet useful fact:
\begin{fact}
\label{lem:tricky}Let $\Pi$ be a renaming, $F$ and $G$ be CSP
instances. Then
\begin{enumerate}
\item $\sigma\in\mathcal{S}\left(F\right)$ iff $\Pi\left(\sigma\right)\in\mathcal{S}\left(\Pi\left(F\right)\right)$;
\item $A_{\Pi\left(F\right)}\left(\Pi\left(\sigma\right),x\right)=\pi_{x}\left(A_{F}\left(\sigma,x\right)\right)$.
\end{enumerate}
\end{fact}
Note that almost all sets of CSP instances we know to be dealt with
in the literature are closed under renaming.

Let $\mathcal{F}$ be some set of instances closed under renaming.
We prove in the sequel that $\gamma\left(\mathcal{F}\right)=\sum_{F\in\mathcal{F}}\left|\mathcal{M}_{\prec}\left(F\right)\right|$
for any instance independent orientation $\prec$ on solutions as
defined in section \ref{sub:Partial-ordering}. That can be interpreted
as follows: on average on $\mathcal{F}$, the weight of all solutions
is equal to the number of minimal solutions, independently of the
orientation $\prec$. The proof idea is to partition the couples (solutions,
instances) in a way that the weight of each class of the partition
has a weight of $1$ and corresponds to a minimal element for $\prec$.

We define the set $\mathcal{C}$ of couples $\left(\sigma,F\right)$
where $F$ is an element of $\mathcal{F}$ and $\sigma$ a solution
of $F$:\begin{eqnarray}
\mathcal{C} & = & \left\{ \left(\sigma,F\right)\right\} _{\substack{F\in\mathcal{F}\\
\sigma\in\mathcal{S}\left(F\right)}
}\enskip.\end{eqnarray}

$\gamma\left(\mathcal{F}\right)$ can be written as \begin{equation}
\gamma\left(\mathcal{F}\right)=\sum_{\left(\sigma,F\right)\in\mathcal{C}}W_{F}\left(\sigma\right)\enskip.\end{equation}

For some variable $x$ and some valuations $v_{1}$ and $v_{2}$,
we define the permutation $\pi_{x,v_{1},v_{2}}$ on $D$ as the transposition
which swaps $v_{1}\left(x\right)$ and $v_{2}\left(x\right)$, and
the renaming $\Pi_{v_{1},v_{2}}$ as the collection of these permutations,
variable per variable:\begin{eqnarray}
\pi_{x,v_{1},v_{2}}\left(a\right) & = & \begin{cases}
v_{1}\left(x\right) & \mbox{if }a=v_{2}\left(x\right)\\
v_{2}\left(x\right) & \mbox{if }a=v_{1}\left(x\right)\\
a & \mbox{otherwise}\end{cases}\enskip;\\
\Pi_{v_{1},v_{2}} & = & \left(\pi_{x,v_{1},v_{2}}\right)_{x\in X}\enskip.\end{eqnarray}

Note that these definitions are symmetric in $v_{1}$ and $v_{2}$.
Moreover note that $\Pi_{v_{1},v_{2}}\left(v_{1}\right)=v_{2}$ and
$\Pi_{v_{1},v_{2}}\left(v_{2}\right)=v_{1}$.

Consider a formula $F\in\mathcal{F}$ and a solution $\tau$ of $F$.
We denote by $\chi_{F}\left(\tau\right)$ the set of valuations $\sigma$
assigning each variable $x$ one of the values in the set $A_{F}\left(\tau,x\right)$:

\begin{eqnarray}
\chi_{F}\left(\tau\right) & = & \prod_{x\in X}A_{F}\left(\tau,x\right)\enskip.\end{eqnarray}

When $\tau$ is a solution of $F$, we denote by $C\left(\tau,F\right)$
the set of all renamings of $\left(\tau,F\right)$ ranging in $\chi_{F}\left(\tau\right)$: 

\begin{equation}
C\left(\tau,F\right)=\left\{ \left(\sigma,\Pi_{\sigma,\tau}\left(F\right)\right)\right\} _{\sigma\in\chi_{F}\left(\tau\right)}\enskip.\end{equation}

\begin{lem}
\label{lem:conservation-of-A}If $\tau$ is a solution of $F$ and
$\sigma\in\chi_{F}\left(\tau\right)$, then $\sigma$ is a solution
of $G=\Pi_{\sigma,\tau}\left(F\right)$ and for all variable $x\in X$,
$A_{G}\left(\sigma,x\right)=A_{F}\left(\tau,x\right)$.\end{lem}
\begin{proof}
Since $G=\Pi_{\sigma,\tau}\left(F\right)$, $\sigma=\Pi_{\sigma,\tau}\left(\tau\right)$
and $\tau$ is a solution of $F$, by fact \ref{lem:tricky} we know
that $\sigma$ is a solution of $G$. Moreover by fact \ref{lem:tricky},
for every variable $x$, $A_{G}\left(\sigma,x\right)=\pi_{x,\sigma,\tau}\left(A_{F}\left(\tau,x\right)\right)$.
By definition of $\chi_{F}$, $\sigma\left(x\right)\in A_{F}\left(\tau,x\right)$.
Since $\pi_{x,\sigma,\tau}$ swaps two values $\tau\left(x\right)$
and $\sigma\left(x\right)$ that are both elements of $A_{F}\left(\tau,x\right)$,
$\pi_{x,\sigma,\tau}\left(A_{F}\left(\tau,x\right)\right)=A_{F}\left(\tau,x\right)$,
hence $A_{G}\left(\sigma,x\right)=A_{F}\left(\tau,x\right)$.\end{proof}
\begin{lem}
\label{lem:partition-of-C}The set $\left\{ C\left(\tau,F\right)\right\} _{\substack{F\in\mathcal{F}\\
\tau\in\mathcal{M}_{\prec}\left(F\right)}
}$ is a partition of $\mathcal{C}$.\end{lem}
\begin{proof}
If $\left(\sigma,G\right)\in C\left(\tau,F\right)$, then by lemma
\ref{lem:conservation-of-A}, $\sigma$ is a solution of $G$. Moreover,
by closure of $\mathcal{F}$ under renaming, $G\in\mathcal{F}$. Thus
$C\left(\tau,F\right)\subseteq\mathcal{C}$. Now it is sufficient
to prove that $\forall\left(\sigma,G\right)\in\mathcal{C}$ there
exists a unique $\left(\tau,F\right)$ where $F\in\mathcal{F}$, $\tau$
is a minimal solution of $F$ and $\left(\sigma,G\right)\in C\left(\tau,F\right)$.
\begin{itemize}
\item Existence of $\left(\tau,F\right)$: for every $x$, let $\tau\left(x\right)$
be the minimal value in $A_{G}\left(\sigma,x\right)$ according to
the order $<_{x}$ underlying $\prec$; by construction, $\tau\in\chi_{G}\left(\sigma\right)$.
Consider the renaming $\Pi_{\sigma,\tau}$ and let $F=\Pi_{\sigma,\tau}\left(G\right)$.
By lemma \ref{lem:conservation-of-A}, $\tau$ is a solution of $F$
and for all variable $x$, $A_{F}\left(\tau,x\right)=A_{G}\left(\sigma,x\right)$.
Since for all $x\in X$, $\tau\left(x\right)$ is the minimal value
in $A_{F}\left(\tau,x\right)$, $\tau$ is minimal for the orientation
$\prec$. Moreover for all $x$, $\sigma\left(x\right)\in A_{F}\left(\tau,x\right)$,
thus $\sigma\in\chi_{F}\left(\tau\right)$; and since $G=\Pi_{\sigma,\tau}\left(F\right)$,
we have that $\left(\sigma,G\right)\in C\left(\tau,F\right)$. 
\item Uniqueness of $\left(\tau,F\right)$: let $\left(\tau',F'\right)$
be such that $C\left(\tau',F'\right)\ni\left(\sigma,G\right)$, i.e.
$\sigma\in\chi_{F}\left(\tau'\right)$ and $G=\Pi_{\sigma,\tau'}\left(F'\right)$;
then by lemma \ref{lem:conservation-of-A}, for all variable $x$,
$A_{G}\left(\sigma,x\right)=A_{F'}\left(\tau',x\right)$. By minimality
of $\tau'$, $\tau'\left(x\right)$ must be the minimum of $A_{G}\left(\sigma,x\right)$
for each variable $x$.
\end{itemize}
\end{proof}
\begin{lem}
\label{lem:inner-sum} Suppose that the weight $w_{F}$ is obtained
from a unitary and instance independent generator $\omega$. Let $\left(\tau,F\right)$
be an element of $\mathcal{C}$; then $\sum_{\left(\sigma,G\right)\in C\left(\tau,F\right)}W_{G}\left(\sigma\right)=1$.\end{lem}
\begin{proof}
First note that by lemma \ref{lem:conservation-of-A}, for all $\left(\sigma,G\right)\in C\left(\tau,F\right)$,
we have $A_{G}\left(\sigma,x\right)=A_{F}\left(\tau,x\right)$. Thus:\begin{eqnarray*}
\sum_{\left(\sigma,G\right)\in C\left(\tau,F\right)}W_{G}\left(\sigma\right) & = & \sum_{\left(\sigma,G\right)\in C\left(\tau,F\right)}\prod_{x\in X}w_{G}\left(\sigma,x\right)\\
 & = & \sum_{\left(\sigma,G\right)\in C\left(\tau,F\right)}\prod_{x\in X}\omega\left(x,\sigma\left(x\right),A_{G}\left(\sigma,x\right)\right)\mbox{ since \ensuremath{\omega}\,\ is instance independent}\\
 & = & \sum_{\sigma\in\chi_{F}\left(\tau\right)}\prod_{x\in X}\omega\left(x,\sigma\left(x\right),A_{F}\left(\tau,x\right)\right)\mbox{ since \ensuremath{A_{G}\left(\sigma,x\right)=A_{F}\left(\tau,x\right)}}\\
 & = & \prod_{x\in X}\sum_{\sigma\left(x\right)\in A_{F}\left(\tau,x\right)}\omega\left(x,\sigma\left(x\right),A_{F}\left(\tau,x\right)\right)\\
 & = & \prod_{x\in X}1\mbox{ since \ensuremath{\omega}\,\ is unitary}\\
 & = & 1\enskip.\end{eqnarray*}

\end{proof}

\begin{thm}
\label{thm:set-equality}Let $\mathcal{F}$ be a set of CSP instances
which is closed under renaming. Let $w_{F}$ be a weighting system
built from a unitary and instance independent weight generator $\omega$.
Let $\prec$ be an instance independent orientation. Then it holds
that $\sum_{F\in\mathcal{F}}\left|\mathcal{M}_{\prec}\left(F\right)\right|=\gamma\left(\mathcal{F}\right)$.\end{thm}
\begin{proof}
It is a mere combination of lemmas \ref{lem:partition-of-C} and \ref{lem:inner-sum}:

\begin{eqnarray*}
\gamma\left(\mathcal{F}\right) & = & \sum_{F\in\mathcal{F}}\gamma\left(F\right)\\
 & = & \sum_{F\in\mathcal{F}}\sum_{\sigma\in\mathcal{S}\left(F\right)}W_{F}\left(\sigma\right)\\
 & = & \sum_{\left(\sigma,F\right)\in C}W_{F}\left(\sigma\right)\\
 & = & \sum_{\substack{F\in\mathcal{F}\\
\tau\in\mathcal{M}_{\prec}\left(F\right)}
}\sum_{\left(\sigma,G\right)\in C\left(\tau,F\right)}W_{G}\left(\sigma\right)\mbox{ by lemma \ref{lem:partition-of-C}}\\
 & = & \sum_{\substack{F\in\mathcal{F}\\
\tau\in\mathcal{M}_{\prec}\left(F\right)}
}1\mbox{ by lemma \ref{lem:inner-sum}}\\
 & = & \sum_{F\in\mathcal{F}}\left|\mathcal{M}_{\prec}\left(F\right)\right|\enskip.\end{eqnarray*}

\end{proof}

Closure under renaming involves symmetry, so it is not surprising
that on average all weightings on the one hand and all orderings on
the other hand should be equivalent. What is more surprising though,
is the fact that weightings and orderings are equivalent. This is
noteworthy because weights are simpler to handle in calculations (they
yield more compact and tractable formulas, see \citet{Boufkhad2010}).

\section{Conclusion and Perspectives}

Through our Weight Conservation Theorem we gave sufficient conditions
to have a correct weighting on solutions of CSPs. We were able to
apply it to two different weightings: the first one, which is very
general, was built from seeds and dispatchers; the second one was
specifically designed to improve on Maneva et al.'s weighting. Thanks
to this new weighting scheme, we obtained an improvement on the upper
bound on the existence of non-trivial cores obtained by Maneva et
al. to $4.419$.

We also showed an equivalence between weighting and ordering over
a set closed under renaming when they are instance independent. On
the contrary, when weighting and ordering may depend on instances,
we showed that given an homogeneous weighting it is possible to find
an ordering which is not worse, but what happens for heterogeneous
weightings? is it always possible to find for a given weighting a
corresponding ordering?

Other perspectives include: is it possible to define a correct non-uniform
weighting? how to generalize boolean partial valuations to general
CSPs? how to extend our weighting when considering neighbors of neighbors,
or more generally neighbors at bounded distance?

\bibliographystyle{elsart-num-names}
\bibliography{library}

\appendix

\section{Proof of Theorem \ref{thm:4.419}}

\subsection{First Moment of Cores }

We strongly advise the reader to read \citet{Maneva2008} before reading
our calculations, since we reuse all of the notations and arguments
from there. So we only highlight the similarities and the differences.

The calculation of \citet{Maneva2008} works in two steps:
\begin{enumerate}
\item compute an upper bound function on the first moment $f$ of covers
and discard the range of variables where $f<0$ (because if there
are no covers, then there are no cores either);
\item compute an upper bound function of the first moment $f+h$ of cores
and maximize it on the remaining domain of the variables.
\end{enumerate}
In particular they introduce some variables $s,t,u$; to these we
add $v$:
\begin{description}
\item [{$s$:}] the size of the cover or core i.e. the number of variables
in a controlled self-constrained set (they get a weight of $1$);
we denote them by symbols $x_{i}$ where $i\in SC=\{1,...,s\}$; 
\item [{$v$:}] number of invertible variables (they get a weight of $\frac{1-\rho}{2}$);
we denote them by symbols $x_{i}$ where $i\in I=\{s+1,...,s+t+v\}$;
\item [{$u-v$:}] proportion of starrable but non-invertible variables
(they get a weight of $1-\rho$); we denote them by symbols $x_{i}$
where $i\in SNI=\{s+t+v+1,...,s+u\}$;
\item [{$t-u$:}] number of non-starrable variables not in the previous
self-constrained set (they get a weight of $1$); we denote them by
symbols $x_{i}$ where $i\in NS=\{s+u+1,...,s+t\}$;
\item [{$n-s-t$:}] number of variables at value $*$ (they get a weight
of $\rho$); we denote them by symbols $x_{i}$ where $i\in S=\{s+1,...,s+t-u\}$.
\end{description}
$p$ is the probability for a clause of type 3 to be included in the
Poisson model. In table \ref{tab:Clauses} we sum up all possible
types of clauses in order to count them. We assume by symmetry that
we have an assignment with values in $\left\{ 0,*\right\} $ only
(no $1$'s). Note that clauses of type 1 and 2 are the same as in
\citet{Maneva2008} and are used in the expression of $f$ rather
than in the expression of $h$.

\begin{table}
\caption{\label{tab:Clauses}Clauses and their sizes.}

\centering{}\begin{tabular}{|c|c|c|c|c|}
\hline 
Types & Clause  & Sets of subscripts & Size & Status\tabularnewline
\hline
 & $x_{i}\vee x_{j}\vee x_{k}$ & $i,j,k\in SC$ & ${s \choose 3}$ & forbidden\tabularnewline
\cline{2-5} 
1 & $x_{i}\vee x_{j}\vee x_{k}$  & $i,j\in SC$, $k\in V\backslash SC$  & $2\left(n-s\right){s \choose 2}$ & forbidden\tabularnewline
 &  $x_{i}\vee x_{j}\vee\overline{x_{k}}$ &  &  & \tabularnewline
\hline 
2 & $x_{i}\vee x_{j}\vee\overline{x_{k}}$ & $i,j,k\in SC$ & ${s \choose 2}$  & At least one\tabularnewline
 &  &  & for each $x_{k}$ & for each $k\in SC$\tabularnewline
\hline
 & $x_{i}\vee x_{j}\vee x_{k}$ & $i\in SC\cup NS\cup I\cup SNI$,  & ${t \choose 3}+s{t \choose 2}$ & forbidden\tabularnewline
 &  & $j,k\in NS\cup I\cup SNI$ &  & \tabularnewline
\cline{2-5} 
 & $x_{i}\vee x_{j}\vee x_{k}$ & $i\in SC\cup NS\cup I\cup SNI$, & $2\left(n-s-t\right)\left({t \choose 2}+st\right)$ & forbidden\tabularnewline
 &  $x_{i}\vee x_{j}\vee\overline{x_{k}}$ & $j\in NS\cup I\cup SNI$  &  & \tabularnewline
 &  & $k\in S$ &  & \tabularnewline
\cline{2-5} 
3 & $x_{i}\vee x_{j}\vee\overline{x_{k}}$ & $i\in SC\cup NS\cup I\cup SNI$,  & $u\left({t \choose 2}+st\right)$ & forbidden\tabularnewline
 &  & $j\in NS\cup I\cup SNI$  &  & \tabularnewline
 &  & $k\in I\cup SNI$ &  & \tabularnewline
\cline{2-5} 
 & $ $$x_{i}\vee\overline{x_{j}}\vee x_{k}$  & $i\in SC\cup NS\cup I\cup SNI$ & $2\left(n-s-t\right)v\left(s+t\right)$ & forbidden\tabularnewline
 &  $x_{i}\vee\overline{x_{j}}\vee\overline{x_{k}}$ & $j\in I$  &  & \tabularnewline
 &  &  $k\in S$ &  & \tabularnewline
\cline{2-5} 
 & $x_{i}\vee x_{j}\vee\overline{x_{k}}$ & $i\in SC\cup NS\cup I\cup SNI$ & ${t \choose 2}+st$  & At least one \tabularnewline
 &  & $j\in NS\cup I\cup SNI$  & for each $x_{k}$ & for each $k\in NS$\tabularnewline
 &  &  $k\in NS$ &  & \tabularnewline
\cline{2-5} 
 & $x_{i}\vee\overline{x_{j}}\vee x_{k}$  & $i\in SC\cup NS\cup I\cup SNI$ & $2\left(n-s-t\right)\left(s+t\right)$ & At least one \tabularnewline
 &  $x_{i}\vee\overline{x_{j}}\vee\overline{x_{k}}$ & $j\in SNI$  & for each $x_{j}$ & for each $j\in SNI$\tabularnewline
 &  &  $k\in S$ &  & \tabularnewline
\hline
\end{tabular}
\end{table}

As \citet{Maneva2008} did, we first define the following quantity:\begin{eqnarray*}
Q & = & \left(1-p\right)^{{t \choose 3}+s{t \choose 2}+2\left(n-s-t\right)\left({t \choose 2}+st\right)+u\left({t \choose 2}+st\right)+2v\left(n-s-t\right)\left(s+t\right)}\\
 &  & \cdot\left(1-\left(1-p\right)^{{t \choose 2}+st}\right)^{t-u}\left(1-\left(1-p\right)^{2\left(n-s-t\right)\left(s+t\right)}\right)^{u-v}\enskip.\end{eqnarray*}

\newpage

Now with $Q$ we can write down the first moment of cores having $n-s-t$
variables at value $*$:

\begin{eqnarray*}
\mathrm{E}Z_{t} & = & \rho^{n-s-t}{n-s \choose t}2^{t}\sum_{u=0}^{t}\left(1-\rho\right)^{u}{t \choose u}\sum_{v=0}^{u}2^{-v}{u \choose v}Q\\
 & = & \rho^{n-s-t}{n-s \choose t}2^{t}\left(1-p\right)^{{t \choose 3}+s{t \choose 2}+2\left(n-s-t\right)\left({t \choose 2}+st\right)}\\
 &  & \cdot\sum_{u=0}^{t}\left(1-\rho\right)^{u}{t \choose u}\left(1-p\right)^{u\left({t \choose 2}+st\right)}\left(1-\left(1-p\right)^{{t \choose 2}+st}\right)^{t-u}\\
 &  & \cdot\sum_{v=0}^{u}2^{-v}{u \choose v}\left(1-p\right)^{2v\left(n-s-t\right)\left(s+t\right)}\left(1-\left(1-p\right)^{2\left(n-s-t\right)\left(s+t\right)}\right)^{u-v}\\
 & = & \rho^{n-s-t}{n-s \choose t}2^{t}\left(1-p\right)^{{t \choose 3}+s{t \choose 2}+2\left(n-s-t\right)\left({t \choose 2}+st\right)}\\
 &  & \cdot\sum_{u=0}^{t}\left(1-\rho\right)^{u}{t \choose u}\left(1-p\right)^{u\left({t \choose 2}+st\right)}\left(1-\left(1-p\right)^{{t \choose 2}+st}\right)^{t-u}\left(1-\frac{\left(1-p\right)^{2\left(n-s-t\right)\left(s+t\right)}}{2}\right)^{u}\\
 & = & \rho^{n-s-t}{n-s \choose t}2^{t}\left(1-p\right)^{{t \choose 3}+s{t \choose 2}+2\left(n-s-t\right)\left({t \choose 2}+st\right)}\\
 &  & \cdot\left(\left(1-\rho\right)\left(1-p\right)^{{t \choose 2}+st}\left(1-\frac{\left(1-p\right)^{2\left(n-s-t\right)\left(s+t\right)}}{2}\right)+1-\left(1-p\right)^{{t \choose 2}+st}\right)^{t}\\
 & = & \rho^{n-s-t}{n-s \choose t}2^{t}\left(1-p\right)^{{t \choose 3}+s{t \choose 2}+2\left(n-s-t\right)\left({t \choose 2}+st\right)}\\
 &  & \cdot\left(1-\left(1-p\right)^{{t \choose 2}+st}\left(\rho+\frac{\left(1-\rho\right)}{2}\left(1-p\right)^{2\left(n-s-t\right)\left(s+t\right)}\right)\right)^{t}\enskip.\end{eqnarray*}

$s=\lfloor an\rfloor$, $t=bn$ and $p=\frac{3\alpha\left(1-d\right)}{n^{2}\left(4-a^{2}\left(3-a\right)\right)}+o\left(\frac{1}{n^{2}}\right)$,
so:

\begin{eqnarray*}
h=\lim_{n\to+\infty}\frac{\ln\mathrm{E}Z_{t}}{n} & = & \ln\left(\frac{\left(1-a\right)^{1-a}2^{b}\rho^{1-a-b}}{b^{b}\left(1-a-b\right)^{1-a-b}}\right)-\frac{\alpha\left(1-d\right)b\left(b\left(6-5b-3a\right)+12\left(1-a-b\right)a\right)}{2\left(4-a^{2}\left(3-a\right)\right)}\\
 &  & +b\ln\left(1-e^{-A}\left(\rho+\frac{1-\rho}{2}e^{-B}\right)\right)\enskip.\end{eqnarray*}

where $A=\frac{3\alpha\left(1-d\right)b\left(b+2a\right)}{2\left(4-a^{2}\left(3-a\right)\right)}$
and $B=\frac{6\alpha\left(1-d\right)\left(1-a-b\right)\left(a+b\right)}{\left(4-a^{2}\left(3-a\right)\right)}$.

\subsection{Maximization}

Just like \citet{Maneva2008}, we want to show that for all $\alpha\in\left[4.419,4.453\right]$,
$a\in\left[\frac{1}{4.453e^{2}},1\right]$ and $r>1$, when $\rho\left(a\right)=0.3758a+0.7067$,
either $f\left(\alpha,a,r\right)<0$ or for all $b\in\left[0,1-a\right]$,
$f\left(\alpha,a,r\right)+h\left(\alpha,a,r,\rho\left(a\right),b\right)<0$.
Note that our function $f$ is the same as in \citet{Maneva2008};
only $h$ differs.

As in \citet{Maneva2008}, if $r<1.2$ then $\frac{\partial f}{\partial r}>0$,
and if $r>670$ then $\frac{\partial f}{\partial r}<0$, so we are
left with the region $a\in\left[\frac{1}{4.453e^{2}},0.999\right]$,
$r\in\left[1.2,670\right]$. The points where $f\left(4.419,a,r\right)>-0.0001$
are depicted in figure \ref{fig:banana}. This corresponds to the
domain where we must check that $f+h$ is negative. We get bounds
slightly different from \citet{Maneva2008}: $a\in\left[0.28,0.75\right]$,
$r\in\left[1.4,14\right]$, $\alpha\in\left[4.419,4.453\right]$ and
$b\in\left[0,1-a\right]$.

Figure \ref{fig:courbes-f-f+h} gives the shapes of our $f$ and $f+h$
at $\alpha=4.419$ and at $\alpha=4.453$.

\begin{figure}
\begin{centering}
\includegraphics[width=8cm]{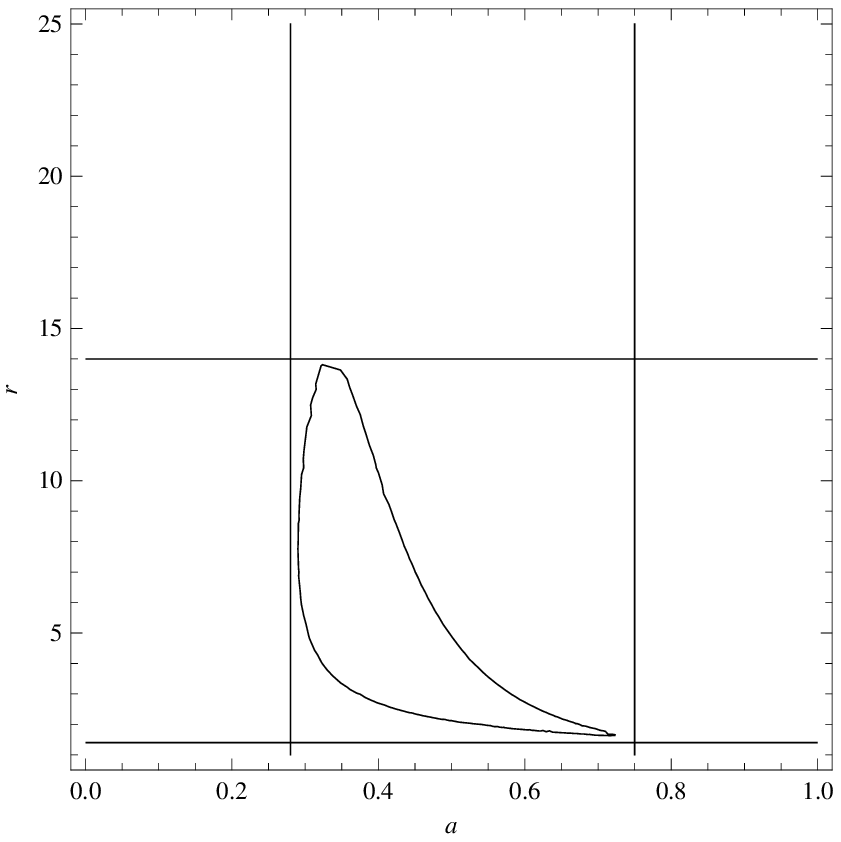}
\par\end{centering}

\caption{\label{fig:banana}$f\left(4.419,a,r\right)>-0.0001$ inside the contour
line, so only when $a\in\left[0.28,0.75\right]$ and $r\in\left[1.4,14\right]$. }

\end{figure}

\begin{figure}
\begin{centering}
\includegraphics[width=0.9\columnwidth]{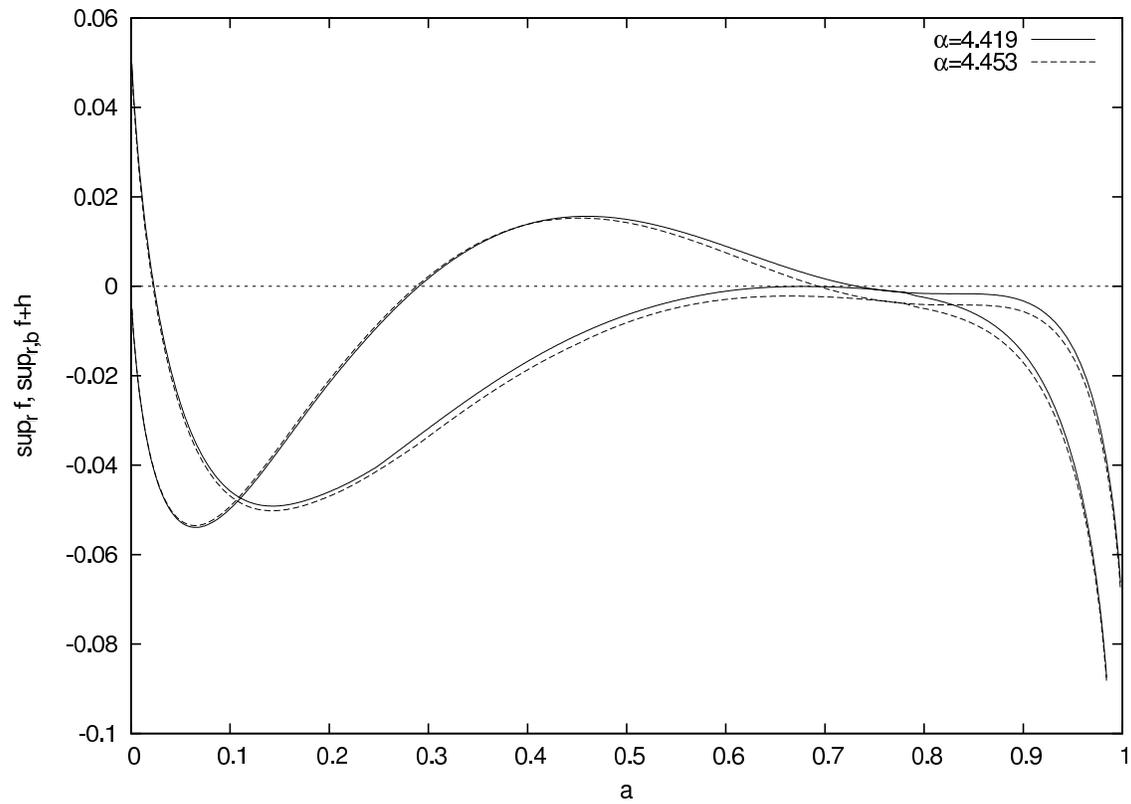}
\par\end{centering}

\caption{\label{fig:courbes-f-f+h}$f$ and $f+h$ at $\alpha=4.453$ and at
$\alpha=4.419$ (with our weights). }

\end{figure}

The domain of the variables is a finite product of segments, all functions
involved are smooth, except at the boundary points where $b=1-a$
(but this is due to the asymptotic equivalent we used for the binomial
coefficient). So in order to maximize $f+h$, following again \citet{Maneva2008},
we performed a sweep over this domain with a step of $0.001$ on all
variables (and a step of $10^{-5}$ in the vicinity of the maximum).
In the end we checked our result using the \texttt{FindMaximum} function
of \texttt{Mathematica}. The maximum of $f+h$ is $-0.0000277225$,
and the values of the variables at this point are $\alpha=4.419$,
$a=0.678206$, $b=0.0299196$ and $r=1.79833$.
\end{document}